\documentclass[11pt,english]{article}
\usepackage[T1]{fontenc}
\usepackage[latin9]{inputenc}
\usepackage{float}
\usepackage{textcomp}
\usepackage{amsmath}
\usepackage{amssymb}
\usepackage{graphicx}
\usepackage{algorithm}
\usepackage[noend]{algorithmic}
\usepackage{caption}
\usepackage{subcaption}

\makeatletter

\floatstyle{ruled}
\newfloat{algorithm}{tbp}{loa}
\providecommand{\algorithmname}{Algorithm}
\floatname{algorithm}{\protect\algorithmname}

\usepackage{latexsym}\@ifundefined{definecolor}
 {\usepackage{color}}{}
\usepackage{mathtools}\usepackage{amsthm}\usepackage{epsfig}\usepackage{algorithmic}\usepackage{wasysym}\usepackage{wrapfig}\usepackage{appendix}\usepackage{fancyhdr}\usepackage{enumerate}
\usepackage{times,euscript}

\usepackage{fullpage}

\def\reals{{\mathbb R}}
\def\eps{{\varepsilon}}

\def\bd{{\partial}}

\def\dFrechet{{\delta^{*}}} 
\def\dsoneFrechet{{\delta^{-}(P,Q)}} 
\def\dstwoFrechet{{\delta^{+}(P,Q)}} 
\def\scsFrechet{{\delta^{-}(P,f)}} 

\newtheorem{theorem}{Theorem}[section]
\newtheorem{lemma}[theorem]{Lemma}

\newtheorem{observation}[theorem]{Observation}

\newcommand{\MyFrame}[1]{\noindent \framebox[\textwidth]{ \begin{minipage}{0.97\textwidth} #1 \end{minipage}}}%
\newcommand{\Lim}[1]{\raisebox{0.5ex}{\scalebox{0.8}{$\displaystyle \lim_{#1}\;$}}}

\def\paragraph#1{\medskip\noindent\textbf{#1}}

\makeatother
\date{}
\usepackage{babel}
\begin{document}
\renewcommand\thepage{}
\begin{titlepage}
\title{The Discrete and Semicontinuous Fr\'echet Distance with Shortcuts \\via Approximate Distance Counting and Selection Techniques\thanks{Work by Omrit Filtser has been partially supported by the Lynn and William Frankel Center for Computer Sciences.
Work by Haim Kaplan has been supported by Israel Science Foundation grant no.\ 822/10 and 1841/14, and the
German-Israeli Foundation for Scientific Research and Development (GIF) grant no.\ 1161/2011, and the Israeli Centers of Research Excellence (I-CORE) program, (Center no.~4/11).
Work by Matya Katz has been partially supported by grant 1045/10 from the Israel Science Foundation, and by grant 2010074 from the United States -- Israel Binational Science Foundation.
Work by Micha Sharir has been supported
by Grant 892/13 from the Israel Science Foundation,
by the Israeli Centers of Research Excellence (I-CORE)
program (Center no.~4/11),
and by the Hermann Minkowski--MINERVA Center for Geometry at Tel Aviv
University.
Work by Micha Sharir and Rinat Ben Avraham has been supported
by Grant 2012/229 from the U.S.-Israeli Binational Science Foundation.
A preliminary version of this paper has appeared in \textit{Proc. 30th Annu. Sympos. Computational Geometry} (2014), 377.}}

\author{Rinat Ben Avraham%
\thanks{School of Computer Science, Tel Aviv University, Tel Aviv 69978, Israel;
\texttt{rinatba@gmail.com}%
} \and Omrit Filtser%
\thanks{Department of Computer Science, Ben-Gurion University, Beer-Sheva 84105, Israel; \texttt{omritna@post.bgu.ac.il}%
} \and Haim Kaplan%
\thanks{School of Computer Science, Tel Aviv University, Tel Aviv 69978, Israel;
\texttt{haimk@post.tau.ac.il}%
} \and Matthew J. Katz%
\thanks{Department of Computer Science, Ben-Gurion University, Beer-Sheva 84105, Israel; \texttt{matya@cs.bgu.ac.il}%
} \and Micha Sharir%
\thanks{School of Computer Science, Tel Aviv University, Tel~Aviv 69978,
Israel; \texttt{michas@post.tau.ac.il }%
} }

\maketitle

\vspace{-1cm}

\begin{abstract}
The \emph{Fr\'echet distance} is a well studied similarity measures between curves. The \emph{discrete Fr\'echet
distance} is an analogous similarity measure, defined for a sequence $P$ of $m$ points and a sequence $Q$ of $n$ points, where the points are usually sampled from input curves.
In this paper we consider a variant, called the \emph{discrete
Fr\'echet distance with shortcuts}, which captures the similarity
between (sampled) curves in the presence of outliers. For the
\emph{two-sided} case, where shortcuts are allowed in both curves,
we give an $O((m^{2/3}n^{2/3}+m+n)\log^3 (m+n))$-time algorithm for
computing this distance. When shortcuts are allowed only in one
noise-containing curve, we give an even faster randomized algorithm
that runs in $O((m+n)^{6/5+\eps})$ time in expectation and with high probability, for any $\eps>0$.
These time bounds are interesting since (i) the best bounds known
for the \emph{Fr\'echet distance} and the \emph{discrete Fr\'echet
distance} (without shortcuts) are  quadratic, or slightly subquadratic, despite
extensive research over many years, and (ii) the only known
algorithms for the \emph{continuous} Fr\'echet distance with
shortcuts are super-quadratic or give constant approximation.

Our techniques are novel and may find further applications. One of
the main new technical results is: Given two sets of points $P$ and
$Q$ and an interval $I$, we develop an algorithm that decides
whether the number of pairs $(x,y)\in P\times Q$ whose distance
${\rm dist}(x,y)$ is in $I$, is less than some given threshold $L$.
The running time of this algorithm decreases as $L$ increases. In
case there are more than $L$ pairs of points whose distance is in
$I$, we can get a small sample of pairs that contains a pair at
approximate median distance (i.e., we can approximately ``bisect''
$I$). We combine this procedure with additional ideas to search,
with a small overhead, for the optimal one-sided Fr\'echet distance
with shortcuts, exploiting the fact that this problem has a very fast decision procedure.  We also show
how to apply this technique for approximate distance selection
(with respect to rank), and for computing the semi-continuous Fr\'echet
distance with one-sided shortcuts. In general, the new technique can
apply to optimization problems for which the decision procedure is
very fast but standard techniques like parametric search makes the
optimization algorithm substantially slower.
\end{abstract}

\end{titlepage}
\renewcommand\thepage{\arabic{page}}
\section{Introduction}

\label{sec:introduction}


Consider a person and a dog connected by a leash, each walking along
a curve from its starting point to its end point. Both are allowed
to control their speed but they cannot backtrack. The {\em Fr\'echet
distance} between the two curves is the minimum length of a leash
that is sufficient for traversing both curves in this manner. The
\emph{discrete fr\'echet distance} replaces the curves by two
sequences of points $P = (p_0,\ldots, p_{m-1})$ and
$Q=(q_0,\ldots,q_{n-1})$, and replaces the person and the dog by two
frogs, the $P$-frog and the $Q$-frog, initially placed at $p_{0}$
and $q_{0}$, respectively. At each move, the $P$-frog or the
$Q$-frog (or both) jumps from its current point to the next. The
frogs are not allowed to backtrack. We are interested in the minimum
length of a ``leash'' that connects the frogs and allows the
$P$-frog and the $Q$-frog to get to $p_{m-1}$ and $q_{n-1}$, respectively.
More formally, for a given length $\delta$ of the leash, a
jump is allowed only if the distances between the two frogs before
and after the jump are both at most $\delta$; the \emph{discrete
Fr\'echet distance} between $P$ and $Q$, denoted by
$\delta_F^{*}(P,Q)$, is then the smallest $\delta>0$ for which there
exists a sequence of jumps that brings the frogs to $p_{m-1}$ and
$q_{n-1}$, respectively.

The Fr\'echet distance and the discrete Fr\'echet distance are used
as similarity measures between curves and sampled curves,
respectively, in many applications. Among these are speech
recognition~\cite{KHMTC98}, signature verification~\cite{MP99},
matching of time series in databases~\cite{KKS05}, map-matching of
vehicle tracking data ~\cite{BPSW05, CDGNW11, WSP06}, and analysis
of moving objects~\cite{BBG08, BBGLL08}.

In many of these applications the curves or the sampled sequences of
points are generated by physical sensors, such as GPS. These sensors
may generate inaccurate measurements, which we refer to as {\em outliers}.
The Fr\'echet distance and the discrete Fr\'echet distance are bottleneck
(min-max) measures, and are therefore sensitive to outliers, and
may fail to capture the similarity between the curves when there are
outliers, because the large distance from an outlier to the other curve
might determine the Fr\'echet distance, making it much larger than
the distance without the outliers.

In order to handle outliers, Driemel
and Har-Peled ~\cite{DH12} introduced the (continuous) Fr\'echet
distance with shortcuts. They considered polygonal curves and
allowed (only) the dog to take shortcuts by walking from a vertex $v$
to any succeeding vertex $w$ along the straight segment connecting $v$
and $w$. This ``one-sided'' variant allows to ``ignore'' subcurves of
one (noisy) curve which substantially deviate from the other (more
reliable) curve. They gave efficient approximation algorithms for the
Fr\'echet distance in such scenarios; these are reviewed in more
detail later on.

Driven by the same motivation of reducing sensitivity to outliers,
we define two variants of the discrete Fr\'echet distance with
shortcuts. In the one-sided variant, we allow the $P$-frog to jump to
any point that comes later in its sequence, rather than just to the
next point. The $Q$ frog has to visit all the $Q$ points in order, as in the standard discrete Fr\'echet distance problem. However, we add the restriction that only a single frog
is allowed to jump in each move (see below for more details).
As in the standard discrete Fr\'echet distance, such a jump is
allowed only if the distances between the two frogs before and
after the jump are both at most $\delta$. The \emph{one-sided
discrete Fr\'echet distance with shortcuts}, denoted as
$\dsoneFrechet$, is the smallest $\delta>0$ for which
there exists such a sequence of jumps that brings the frogs to
$p_{m-1}$ and $q_{n-1}$, respectively. We also define the
{\em two-sided discrete Fr\'echet distance with shortcuts},
denoted as $\dstwoFrechet$, to be the smallest $\delta>0$ for
which there exists a sequence of jumps, where both frogs are
allowed to skip points as long as the distances between the two frogs before and
after the jump are both at most $\delta$. Here too, we allow only one of the frogs
to jump at each move.

In the (standard) discrete Fr\'echet distance, the frogs can make
simultaneous jumps, each to its next point. In contrast, when
allowing shortcuts, we forbid the frogs from making such simultaneous
jumps. This forces the frog making the jump (standard or shortcut)
to stay close to the other frog while making the move. In a sense this
restriction is the discrete analogue of the requirement in the
continuous case, that the dog, when walking on its shortcut segment,
stays close to the person on the other curve, who does not move
during the shortcut. In the two sided case simultaneous jumps
make the problem degenerate as it is possible for the frogs to jump
from $p_{0}$ and $q_{0}$ straight to $p_{m-1}$ and $q_{n-1}$.

\paragraph{Our results.} In this paper we give efficient algorithms for
computing the discrete Fr\'echet distance with one-sided and two-sided
shortcuts.
The
structure of the  one-sided problem allows to decide whether the distance is
no larger than a given $\delta$ in $O(m+n)$ time, and the challenge
is to search for the optimum, using this fast decision procedure,
with the smallest possible overhead. The naive approach would be to
use the $O((m^{2/3}n^{2/3}+m+n)\log (m+n))$-time distance selection
procedure of \cite{KS97}, which would make the running time
$\Omega((m^{2/3}n^{2/3}+m+n)\log (m+n))$, much higher than the
linear cost of the decision procedure.

To tighten this gap, we develop two algorithms. The first algorithm
finds an interval $(\alpha, \beta]$ that contains $\dsoneFrechet$ and, with high probability, contains only $O(L)$ additional critical distances, for a given parameter $1 \le L \le m+n$.
This algorithm runs in
$O((m+n)^{4/3+\eps}/L^{1/3})$ time, in expectation and with high probability, for any
$\eps>0$.
The second algorithm
searches for $\dsoneFrechet$ in $(\alpha, \beta]$ by simulating the decision procedure in an efficient manner. Here, we use the fact that, as a result of the first algorithm, the simulation encounters only $O(L)$ critical distances with high probability. This algorithm is deterministic and runs in $O((m+n)L^{1/2}\log(m+n))$ time.
As $L$ increases the first
algorithm becomes faster and the second algorithm becomes slower. Choosing
$L$ to balance the two gives us an algorithm for the one-sided
Fr\'echet distance with shortcuts that runs in $O((m+n)^{6/5+\eps})$
time in expectation and with high probability, for any $\eps>0$.

We believe that these algorithms are of independent interest, beyond
the scope of computing the one-sided Fr\'echet distance with shortcuts,
and that they may be applicable to other optimization problems over
pairwise distances.  We give two such additional applications.
The first application is of the first algorithm and it is a rank-based approximation of the $k$th
smallest distance. More specifically,
let $k$ and $L$  be such that $0<k<mn$ and $\sqrt{k}\le L\le k$. We give an
algorithm for finding a distance which is the $\kappa$-th
smallest distance, for some rank $\kappa$ satisfying $k-L \le \kappa \le k+L$, that runs in
$O\left(\frac{mnk}{L^2}\log (m+n) + m +n\right)$ time. If $L^2/k\le m+n$ we can also find such a pair in
$O((m+n)^{4/3+\eps}k^{1/3}/L^{2/3})$ time for any $\eps>0$. This time bound holds in expectation and with high probability.

Our second application is a semi-continuous version of the one-sided
Fr\'echet distance with shortcuts. In this problem $P$ is a sequence
of $m$ points and $f\subseteq\mathbb{R}^{2}$ is a polygonal curve of
$n$ edges. A frog has to jump over the points in $P$, connected by a
leash to a person who walks on $f$. The frog can make shortcuts and
skip points, but the person must traverse $f$ continuously. The frog and the person cannot backtrack. We want to
compute the minimum length of a leash that allows the
frog and the person to get to their final positions in such
a scenario. In Section~\ref{sec:semi_continuous} we present an
algorithm, that runs in time $O((m+n)^{2/3}m^{2/3}n^{1/3}\log(m+n))$ in expectation and with high probability, for this problem. While less efficient than the fully discrete
version, it is still significantly subquadratic.

For the two-sided version we take a different approach. More specifically, we use an implicit compact
representation of all pairs in $P\times Q$ at distance at most
$\delta$ as the disjoint union of complete bipartite cliques \cite{KS97}.
This representation allows us to maintain the pairs reachable by the
frogs with a leash of length at most $\delta$ implicitly and
efficiently. Our algorithm runs in $O((m^{2/3}n^{2/3}+m+n)\log^3
(m+n))$ time and requires $O((m^{2/3}n^{2/3}+m+n)\log (m+n))$
space.

Interestingly, the algorithms developed for these variants of the discrete Fr\'echet distance problem are sublinear in the size of $P\times Q$
and way below the slightly subquadratic bound for the discrete
Fr\'echet distance, obtained in \cite{ABKS12}.

\paragraph{Background.} The Fr\'echet distance and its variants have been
extensively studied in the past two decades.
Alt and Godau~\cite{AG95} showed that the Fr\'echet distance of two
planar polygonal curves with a total of $n$ edges can be computed,
using dynamic programming, in $O(n^{2}\log n)$ time. Eiter and
Mannila~\cite{EM94} showed that the discrete Fr\'echet distance in
the plane can be computed, also using dynamic programming, in
$O(mn)$ time. Buchin et al.~\cite{BBMM12} recently
improved the bound of Alt and Godau and showed how to compute the
Fr\'echet distance in $O(n^2 (\log n)^{1/2} (\log\log n)^{3/2})$
time on a pointer machine, and in $O(n^2 (\log\log n)^2)$ time
on a word RAM \cite{BBMM12}. Agarwal et al.~\cite{ABKS12} showed
how to compute the discrete Fr\'echet distance in
$O\left(\dfrac{nm\log\log n}{\log n}\right)$ time.

As already noted, the (one-sided) continuous Fr\'echet distance
with shortcuts was first studied by Driemel and Har-Peled~\cite{DH12}.
They considered the problem where shortcuts are allowed only between
\emph{vertices} of the noise-containing curve, in the manner outlined
above, and gave approximation algorithms for
solving two variants of this problem. In the first variant,  any
number of shortcuts is allowed, and in the second variant, the
number of allowed shortcuts is at most $k$, for some $k \in
\mathbb{N}$. Their algorithms work efficiently only
for \emph{$c$-packed} polygonal curves; these are curves that
behave ``nicely'' and are assumed to be the input in practice.
Both algorithms compute a
$(3+\varepsilon)$-approximation of the Fr\'echet distance with
shortcuts between two $c$-packed polygonal curves and both run in
near-linear time (ignoring the dependence on $\varepsilon$).
Buchin et al.~\cite{BDS13} consider a more general version of the
(one-sided) continuous Fr\'echet distance with shortcuts, where
shortcuts are allowed between any pair of points of the noise-containing
curve. They show that this problem is NP-Hard. They also give a 3-approximation algorithm for the decision version of this problem
that runs in $O(n^3 \log n)$ time.

We also note that there have been several other works that treat outliers
in different ways. One such result is of Buchin et al.~\cite{BBW09}, who considered the partial Fr\'echet similarity problem. In this problem, given two curves $f$ and $g$, and a distance threshold $\delta$, the goal is to maximize the total length of the portions of
$f$ and $g$ that are matched (using the Fr\'echet distance) with $L_p$ distance smaller than $\delta$.
They gave an algorithm that solves this problem in $O(mn(m + n) \log(mn))$ time, under the $L_1$ or $L_\infty$ norm.
Practical implementations of Fr\'echet distance algorithms, that are made for experiments on real data in map matching applications, remove outliers from the data set ~\cite{CDGNW11,WSP06}.
In another map matching application, Brakatsoulas et
al.~\cite{BPSW05} define the notion of integral Fr\'echet distance to deal with outliers. This distance measure averages over
certain distances instead of taking the maximum.

\section{Preliminaries}
We now give a formal definition of the discrete Fr\'echet distance and its variants.

Let $P=(p_{0},\ldots,p_{m-1})$ and $Q=(q_{0},\ldots,q_{n-1})$ be two
sequences of $m$ and $n$ points, respectively, in the plane. Let $G(V,E)$ denote a graph whose vertex set is $V$ and edge set is $E$, and let $\|\cdot\|$ denote the Euclidean norm.
Fix a distance $\delta>0$, and define the following three directed graphs $G_\delta=G(P\times Q, E_{\delta})$, $G_\delta^-=G(P\times Q, E_{\delta}^-)$, and $G_\delta^+=G(P\times Q, E_{\delta}^+)$, where
\begin{align*}
E_{\delta}= & \left\{ \Bigr((p_{i},q_{j}),(p_{i+1},q_{j})\Bigl)\:\middle|\,\|p_{i}-q_{j}\|,\;\|p_{i+1}-q_{j}\|\le\delta\right\} \bigcup\\
 & \left\{ \Bigr((p_{i},q_{j}),(p_{i},q_{j+1})\Bigl)\:\middle|\,\|p_{i}-q_{j}\|,\;\|p_{i}-q_{j+1}\|\le\delta\right\},\\
E_{\delta}^{-}= & \left\{ \Bigr((p_{i},q_{j}),(p_{k},q_{j})\Bigl)\:\middle|\, k>i,\,\|p_{i}-q_{j}\|,\;\|p_{k}-q_{j}\|\le\delta\right\} \bigcup\\
 & \left\{ \Bigr((p_{i},q_{j}),(p_{i},q_{j+1})\Bigl)\:\middle|\,\|p_{i}-q_{j}\|,\;\|p_{i}-q_{j+1}\|\le\delta\right\}, \\
E_{\delta}^{+}= & \left\{ \Bigr((p_{i},q_{j}),(p_{k},q_{j})\Bigl)\:\middle|\, k>i,\,\|p_{i}-q_{j}\|,\;\|p_{k}-q_{j}\|\le\delta\right\} \bigcup\\
 & \left\{ \Bigr((p_{i},q_{j}),(p_{i},q_{l})\Bigl)\:\middle|\, l>j,\,\|p_{i}-q_{j}\|,\;\|p_{i}-q_{l}\|\le\delta\right\}.
\end{align*}
For each of these graphs we say
that a position $(p_{i},q_{j})$ is a \emph{reachable}
position if $(p_{i},q_{j})$ is reachable from $(p_{0},q_{0})$ in the respective graph.
Then the discrete Fr\'echet distance (DFD for short) $\dFrechet(P,Q)$ is the smallest $\delta>0$
for which $(p_{m-1},q_{n-1})$ is a reachable position in $G_\delta$.
Similarly, the one-sided Fr\'echet distance with shortcuts (one-sided DFDS for short) $\dsoneFrechet$ is the smallest $\delta>0$
for which $(p_{m-1},q_{n-1})$ is a reachable position in $G_\delta^-$.
Finally, the two-sided Fr\'echet distance with shortcuts (two-sided DFDS for short) $\dstwoFrechet$ is the smallest $\delta>0$
for which $(p_{m-1},q_{n-1})$ is a reachable position in $G_\delta^+$.

\section{Decision procedure for the one-sided DFDS}

\label{sec:DFDS1} We first consider the corresponding decision problem.
That is, given a value $\delta>0$ we wish to decide whether $\dsoneFrechet\le\delta$.

Let $M$ be the matrix whose rows correspond to the elements of $P$ and
whose columns correspond to the elements of $Q$ and $M_{i,j}=1$ if
$\|p_{i}-q_{j}\|\le\delta$, and $M_{i,j}=0$ otherwise. Consider
first the DFD variant (no shortcuts allowed), in which,
at each move, exactly one of the frogs has to jump to the next
point. Suppose that $(p_{i},q_{j})$ is a reachable position of the
frogs. Then, necessarily, $M_{i,j}=1$. If $M_{i+1,j}=1$ then the
next move can be an {\em upward move} in which the $P$-frog moves
from $p_{i}$ to $p_{i+1}$, and if $M_{i,j + 1} = 1$ then the next
move can be a {\em right move} in which the $Q$-frog moves from
$q_{j}$ to $q_{j+1}$. It follows that to  determine whether
$\dFrechet\le\delta$, we need to determine whether there is
a \emph{right-upward staircase} of ones in $M$ that starts at
$M_{0,0}$, ends at $M_{m-1,n-1}$, and consists of a sequence of
interweaving upward moves and right moves (see
Figure~\ref{fig:staircase}(a)).

\begin{figure}[htp]
\centering

\begin{tabular}{ccc}
\includegraphics[scale=0.8]{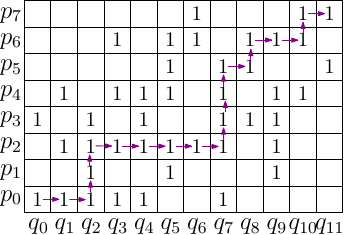} & \hspace{0.5cm}
\includegraphics[scale=0.8]{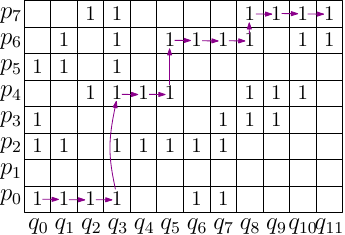} & \hspace{0.5cm}
\includegraphics[scale=0.8]{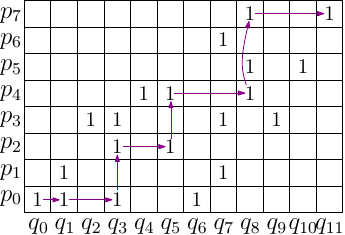}\\
(a) & \hspace{0.5cm}(b) & \hspace{0.5cm}(c)
\end{tabular}

\centering \caption{\small (a) A right-upward staircase (for DFD with no simultaneous jumps). (b) A semi-sparse
staircase (for the one-sided DFDS). (c) A sparse
staircase (for the two-sided DFDS).}
\label{fig:staircase}
\end{figure}

In the
one-sided version of DFDS, given a
reachable position $(p_{i},q_{j})$ of the frogs, the $P$-frog can
move to any point $p_{k},k>i$, for which $M_{k,j}=1$; this is a
\emph{skipping upward move} in $M$ which starts at $M_{i,j}=1$,
skips over $M_{i+1,j},\ldots,M_{k-1,j}$ (some of which may be 0),
and reaches $M_{k,j}=1$. However, in this variant, as in the DFD variant, the $Q$-frog
can only make a \emph{right move}  from $q_{j}$ to $q_{j+1}$,
provided that $M_{i,j+1}=1$ (otherwise no move of the $Q$-frog is
possible at this position). Determining whether
$\dsoneFrechet\le\delta$ corresponds to deciding whether there
is a \emph{semi-sparse staircase} of ones in $M$ that starts at
$M_{0,0}$, ends at $M_{m-1,n-1}$, and consists of an interweaving
sequence of skipping upward moves and (consecutive) right moves (see
Figure~\ref{fig:staircase}(b)).

Assume that $M_{0,0}=1$ and $M_{m-1,n-1}=1$; otherwise, we can immediately
conclude that $\dsoneFrechet>\delta$ and terminate the decision procedure. From now on, whenever we
refer to a semi-sparse staircase, we mean a semi-sparse staircase
of ones in $M$ starting at $M_{0,0}$, as defined above, but without the requirement that it ends at $M_{m-1,n-1}$.

\begin{figure}[htbp]
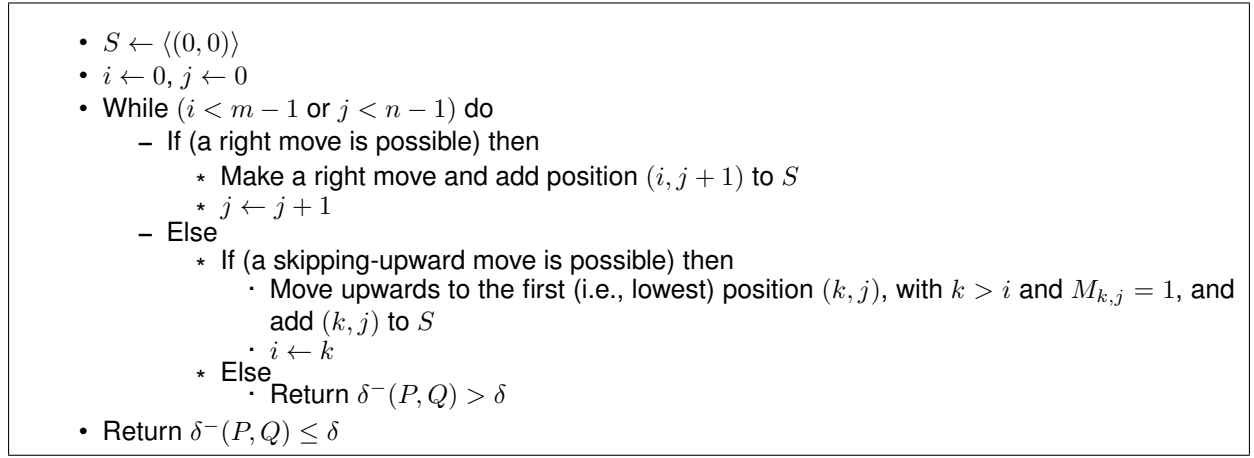

\MyFrame{
\smallskip
\begin{itemize}
{\small \sf
\item $S \leftarrow \langle (0,0)\rangle$ \vspace{-0.3cm}
\item $i \leftarrow 0$, $j \leftarrow 0$ \vspace{-0.3cm}
\item While $(i < m-1$ or $j < n-1)$ do \vspace{-0.3cm}
    \begin{itemize}
    \item {\small \sf If (a right move is possible) then } \vspace{-0.1cm}
        \begin{itemize}
                             \item Make a right move and add position $(i,j+1)$ to $S$ \vspace{-0.1cm}
                             \item $j \leftarrow j+1$ \vspace{-0.2cm}
        \end{itemize}
        \item  Else \vspace{-0.2cm}
        \begin{itemize}
        \item If (a skipping-upward move is possible) then \vspace{-0.1cm}
                            \begin{itemize}
                                   \item Move upwards to the first (i.e., lowest) position $(k,j)$, with $k>i$ and $M_{k,j}=1$, and add $(k,j)$ to $S$ \vspace{-0.1cm}
                             \item $i\leftarrow k$ \vspace{-0.2cm}
                            \end{itemize}

        \item Else \vspace{-0.2cm}
                \begin{itemize}
                     \item Return $\dsoneFrechet>\delta$ \vspace{-0.3cm}
                \end{itemize}
\end{itemize}
\end{itemize}
        \item Return $\dsoneFrechet\le\delta$
}
\end{itemize}
}\caption{Decision procedure for the one-sided discrete Fr\'echet distance with shortcuts.} \label{alg:semi-sparse}
\end{figure}

The algorithm of Figure~\ref{alg:semi-sparse}, that implements the
decision procedure, constructs an upward-skipping path $S$ by always
making a right move if possible. If a right-move is not possible the algorithm makes an upward-skipping move (if possible). The correctness of the decision
procedure is established by the following lemma.

\begin{lemma} \label{lem:one-sided}
If there exists an upward-skipping path that ends at
$(m-1,n-1)$, then $S$ also ends at $(m-1,n-1)$. Hence $S$ ends at $(m-1,n-1)$ if and only if $\dsoneFrechet \le \delta$.
\end{lemma}

\begin{proof}
Let $S'$ be an upward-skipping path that ends at $(m-1,n-1)$. We think
of $S'$ as the sequence of its positions (necessarily $1$-entries) in
$M$. Note that $S'$ has at least one position in each column of $M$,
since skipping is not allowed when moving rightwards. We claim that
for each position $(k,j)$ in $S'$, there exists a position
$(i,j)$ in $S$, such that $i \le k$. This, in particular, implies
that $S$ reaches the last column, and thereby, by the definition of the decision procedure to $(m-1,n-1)$.

We prove the claim by induction on $j$. It clearly holds for
$j=0$ as both $S$ and $S'$ start at $(0,0)$. We assume then
that the claim holds for $j=\ell-1$, and establish it for $\ell$.
That is, assume that if $S'$ contains an entry $(k,\ell-1)$, then $S$
contains $(i,\ell-1)$ for some $i\le k$. Let $(k',\ell)$ be
the lowest position of $S'$ in column $\ell$; clearly, $k'\ge k$.
We must have $M_{k',\ell-1} = 1$  (as (i) by assumption, $S'$ has reached $(k',\ell)$ from the previous column, and (ii) the only way to move from a column to
the next one is by a right move). By the definition of the decision procedure $S$ is extended by a sequence (which may be empty if $M_{i,\ell} = 1$) of skipping
upward moves in column $\ell-1$ until reaching the lowest index $i'\geq i$, for which $M_{i',\ell-1}=1$ and $M_{i',\ell}$ is 1. (This is the lowest instance in which $S$ can be extended by a right move.) But since $M_{k',\ell-1} = 1$ and $M_{k',\ell}=1$, and $i< k'$, we get that $i' \le k'$, as required. (Note that the existence of $k'$ implies that $i'$ is well defined.)
\end{proof}

It is easy to verify that a straightforward implementation of the decision procedure runs in $O(m+n)$ time.

\section{One-sided DFDS optimization via approximate distance counting and selection}
\label{sec:tec}

We now show how to use the decision procedure of Figure~\ref{alg:semi-sparse} to solve the optimization problem of the one-sided discrete Fr\'echet distance with shortcuts.

First note that if we increase $\delta$
continuously, the set of $1$-entries of $M$ can only grow, and
this happens when $\delta$ is a distance between a point of $P$ and a
point of $Q$. Performing a binary search over the $O(mn)$ distances between pairs of points in
$P\times Q$ can be done using the distance selection
algorithm of~\cite{KS97}. This will be the method of choice for the two-sided DFDS problem, treated in Section~\ref{sec:two_sided}. Here however, this procedure, which takes
$O(m^{2/3}n^{2/3}\log^3(m+n))$ time, is rather expensive when compared
to the linear cost of the decision procedure. While solving the
optimization problem in close to linear time is still a challenging
open problem, we improve the running time considerably, using randomization, to
$O((m+n)^{6/5+\eps})$ in expectation, for any $\eps>0$.

Our algorithm is based on two independent building blocks:

\paragraph{Algorithm~\ref{sec:tec}.1}
An algorithm that finds an interval $(\alpha, \beta]$ that contains $\dsoneFrechet$ and, with high probability, contains only $O(L)$ additional critical distances, for a given parameter $1 \le L = L(m,n) \le m+n$ that we fix shortly. This algorithm runs in
$O((m+n)^{4/3+\eps}/L^{1/3})$ time, in expectation and with high probability, for any
$\eps>0$.

\paragraph{Algorithm~\ref{sec:tec}.2}
An algorithm that searches for $\dsoneFrechet$ in $(\alpha, \beta]$ by simulating the decision procedure in an efficient manner. At this stage, we use the fact that the simulation encounters only $O(L)$ critical distances (with high probability, as a consequence of Algorithm~\ref{sec:tec}.1). This algorithm is deterministic and runs in $O((m+n)L^{1/2}\log(m+n))$ time.

To balance the running times of Algorithms~\ref{sec:tec}.1 and~\ref{sec:tec}.2,
we choose $L = (m + n)^{2/5+\eps}$, for another, but still arbitrarily small $\eps > 0$. Then, combining the two algorithms results in an overall optimization algorithm that runs in
$O((m+n)^{6/5+\eps})$ time, in expectation and with high probability, as further elaborated in Section~\ref{sec:overall_shortcuts}.

We describe Algorithm~\ref{sec:tec}.1 in Section~\ref{sec:finding_interval}, describe Algorithm~\ref{sec:tec}.2 in Section~\ref{sec:searching_in_interval}, and combine the algorithms in Section~\ref{sec:overall_shortcuts}. In Section~\ref{sec:correctness}, we prove the correctness and analyze the running times of the algorithms.

We believe that Algorithm~\ref{sec:tec}.1 is of independent interest, and we give another application of it to a different distance-related problem in Section~\ref{sec:cor:aprox_selection}. Independently, we use the same technique for the semicontinuous Fr\'echet distance with one-sided shortcuts, in Section~\ref{sec:semi_continuous}.

\subsection{Algorithm~\ref{sec:tec}.1: Finding an interval that contains $O(L)$ critical distances}\label{sec:finding_interval}
The goal of Algorithm~\ref{sec:tec}.1 is to find an interval $(\alpha,
\beta]$ that contains $\dsoneFrechet$, and $O(L)$ additional distances between pairs of
$P\times Q$. As already noted, we achieve this goal only with high probability.

We start with $(\alpha,\beta]=(0,\infty)$, and iteratively shrink $(\alpha,\beta]$ until it contains $O(L)$ critical distances with high probability. Each iteration consists of three stages.

\paragraph{Stage I.}

We construct, as described below, a batched range counting data structure $\Gamma_L(P,Q,\alpha,\beta)$
for representing some of the pairs $(p,q)\in P \times Q$, as the edge-disjoint union of bipartite
cliques $\{P_{t}\times Q_{t} \mid P_{t}\subseteq P, \; Q_{t}\subseteq
Q\}$.
$\Gamma_L(P,Q,\alpha,\beta)$ consists of two sub-collections of bipartite cliques, $\Gamma_L^1(P,Q,\alpha,\beta)$ and $\Gamma_L^2(P,Q,\alpha,\beta)$.

$\Gamma_L^1(P,Q,\alpha,\beta)$ is a collection of edge-disjoint bipartite cliques, such that if $(p_i,q_j)\in P_{t}\times Q_{t}\in
\Gamma_L^1(P,Q,\alpha,\beta)$ then $|p_i-q_j|\in (\alpha,\beta]$.

$\Gamma_L^2(P,Q,\alpha,\beta)$ is a collection of bipartite cliques that record additional pairs of $P \times Q$. We do not know whether these pairs are in $(\alpha,\beta]$, but we know that all the pairs of $P \times Q$ that are in $(\alpha,\beta]$ are recorded either in $\Gamma_L^1(P,Q,\alpha,\beta)$ or in $\Gamma_L^2(P,Q,\alpha,\beta)$.

$\Gamma_L(P,Q,\alpha,\beta)$ is constructed as follows.
Let $C$ denote the collection
of the circles bounding the  $(\alpha,\beta)$-annuli that are
centered at the points of $P$ (that is, each annulus has inner radius $\alpha$ and outer radius $\beta$). We choose a sufficiently large
constant parameter $1 \leq r \leq m$, and construct a
$(1/r)$-cutting for $C$. That is, for a suitable absolute constant $c$, we
partition the plane into $k\leq cr^2$ cells
$\Delta_1,\ldots,\Delta_k$, each of constant description complexity,
so that each $\Delta_i$ is crossed by at most $m/r$ boundaries of
the annuli, and each $\Delta_i$ contains at most $n/r^2$ points of
$Q$. This can be done deterministically in $O((m+n)r)$ time for any $1\leq r \leq m+n$, as
in~\cite{Cha93,CF90,Mat91}.\footnote{The construction
in~\cite{Cha93,CF90,Mat91} shows that each $\Delta_i$ is crossed by
at most $m/r$ circles in $C$. To ensure that each
$\Delta_i$ contains at most $n/r^2$ points of $Q$, we duplicate each
$\Delta_i$ that contains more than $n/r^2$ points as many times as
needed, and assign to each copy a subset of at most $n/r^2$ of the
points (these sets are pairwise disjoint and cover all the points in the cell).
Then each cell of the resulting subdivision contains at most
$n/r^2$ points, and the size of the cutting is still $O(r^2)$.}
This step captures some of the distances in $(\alpha,\beta]$ --- those between the set $P_{\Delta_i}^C$ of points of $P$ whose annuli fully contain some cell $\Delta_i$ and the set $Q_{\Delta_i}$ of
points of $Q$ contained in $\Delta_i$, for $i=1,\ldots,k$.

However, the number of points of $P$ ($m/r$ points) and the number of points of $Q$ ($n/r^2$ points) that are involved in a cell of the cutting is not balanced.
To balance these numbers, we
now dualize the roles of $P$ and $Q$, in each cell $\Delta_i$
separately, where the set $Q_{\Delta_i}$ of the at most $n/r^2$
points of $Q$ in $\Delta_i$ becomes a set of $(\alpha,\beta)$-annuli
centered at these points, and the set $P_{\Delta_i}$ of the at most
$m/r$ points of $P$ whose annuli boundaries cross $\Delta_i$ is now
regarded as a set of points. We now construct, for each $\Delta_i$,
a $(1/r)$-cutting in this dual setting. We obtain a total of at most
$c^2r^4$ subproblems, each involving at most $m/r^3$ points of $P$
and at most $n/r^3$ points of $Q$.

We output a collection
of complete bipartite graphs, one for each cell either of the primal cutting or of the multiple dual cuttings. For each primal cell $\Delta_i$ we add $P_{\Delta_i}^C\times Q_{\Delta_i}$ to $\Gamma_L^1(P,Q,\alpha,\beta)$, and for each cell $\tau_j$ of a dual cutting associated with some primal cell $\Delta_i$, we add $Q_{\tau_j}^C \times P_{\tau_j}$ to $\Gamma_L^1(P,Q,\alpha,\beta)$, where $Q_{\tau_j}^C$ is the subset of the points of $Q_{\Delta_i}$ whose annuli boundaries contain $\tau_j$, and $P_{\tau_j}$ is the subset of points of $P_{\Delta_i}$ that are contained in $\tau_j$. Note that every containment of a point $q$ of $Q$ in an $(\alpha,\beta)$-annulus centered at a point of $P$ is either stored in (exactly) one of the above complete bipartite graphs, or appears in one of the at most $c^2r^4$ subproblems.

This does not complete the algorithm, and we need to recurse within the cells to produce additional complete bipartite graphs for the desired output.
As just noted, the distances of pairs in $P\times Q$ that lie in $(\alpha,\beta]$ and are not captured by the collection of graphs already in $\Gamma_L^1(P,Q,\alpha,\beta)$ are the distances between centers of annuli
whose boundaries cross some cell $\tau_i$ and points in $\tau_i$ that
lie inside these annuli, over all cells $\tau_i$ of all the dual cuttings. To capture (some of) these distances, we process each of the $O(r^4)$ subproblems recursively (with a primal and dual stages), using the
same parameter $r$. We keep doing so until we get subproblems of size
at most $L$ (in terms of the number of $P$-points plus the number of $Q$-points) and then stop the recursion. At each level of the recursion we add to $\Gamma_L^1(P,Q,\alpha,\beta)$ a collection
of complete bipartite graphs, one for each cell of either the primal or the dual cuttings.
As before, the sets of vertices of the graph associated with a primal or dual cell $\Delta_i$ are the
set of points (of either $P$ or $Q$) whose annuli fully contain $\Delta_i$ and the set of
points (of the other set) contained in $\Delta_i$.

Since we stopped when the size of each subproblem is at most $L$ and did not continue the recursion all the way to problems of constant size, there are distances in $(\alpha,\beta]$ that we did not capture --- those between centers of annuli
whose boundaries cross the cells at the bottom of the recursion and points in those cells that
lie inside these annuli. For each such cell $\Delta$ we add $P_\Delta\times Q_\Delta$ to $\Gamma_L^2(P,Q,\alpha,\beta)$, where $P_\Delta$ are the points of $P$ which are in $\Delta$ and $Q_\Delta$ are the points of $Q$ whose annuli intersect $\Delta$.
Except for these pairs, for which we do not know whether their distances are in $(\alpha,\beta]$, all other pairs with distance in this range are accounted for in the graphs of $\Gamma_L^1(P,Q,\alpha,\beta)$.

This terminates Stage I. As mentioned, $\Gamma_L^1(P,Q,\alpha,\beta)$ and $\Gamma_L^2(P,Q,\alpha,\beta)$ together form the data structure $\Gamma_L(P,Q,\alpha,\beta)$.

In Lemma~\ref{lem:cutting} we prove the following.
The total size of the vertex sets of the bipartite cliques of $\Gamma_L^1(P,Q,\alpha,\beta)$ and $\Gamma_L^2(P,Q,\alpha,\beta)$ is $O((m+n)^{4/3+\eps}/L^{1/3})$, for any $\eps>0$ (the prescribed $\eps$ dictates the choice of $r$). The total number
of pairs in $\Gamma_L^2(P,Q,\alpha,\beta)$ is
$O((m+n)^{4/3+\eps}L^{2/3})$.
$\Gamma_L(P,Q,\alpha,\beta)$ is constructed in overall $O((m+n)^{4/3+\eps}/L^{1/3})$ time. (Note that since $L \le m+n$, $(m+n)^{4/3+\eps}/L^{1/3} = \Omega(m+n)$.)

\paragraph{Stage II.}

Let $S_1\subseteq P\times Q$ (resp., $S_2\subseteq P\times Q$) denote the set of pairs of points corresponding to edges of the bipartite cliques in $\Gamma_L^1(P,Q,\alpha,\beta)$ (resp., $\Gamma_L^2(P,Q,\alpha,\beta)$). Let $S_2'$ denote
 the subset of the pairs $(p,q)$ of $S_2$ such that $|p-q| \in (\alpha,\beta]$.

We determine how many pairs of points are in $S_1$, by counting the number of edges in $\Gamma_L^1(P,Q,\alpha,\beta)$. By construction for every $(p,q)\in S_1$, $|p-q|\in (\alpha,\beta]$.

Our next step aims to approximate how many of the distances between pairs in $S_2$ are in
$(\alpha,\beta]$; i.e., how many pairs of $S_2$ are in $S_2'$. A brute-force counting is too expensive, and we use the following more efficient approach.

We sample a set $R_2=\{(p^1,q^1),(p^2,q^2),\ldots, (p^y,q^y)\}$ of $y=c_2(|S_2|/L)\log (m+n)$ pairs from $S_2$ uniformly at random, for some sufficiently large constant $c_2>0$.  It is straightforward to generate such a sample by picking a pair uniformly from a random bipartite clique $X\in \Gamma_L^2(P,Q, \alpha,\beta)$, where the probability of sampling $X$ is proportional to the number of pairs in $X$. Let $R_2'$ denote the subset of pairs of $R_2$ whose distances are in $(\alpha,\beta]$. We compute $|R_2'|$ in a brute-force manner, in $O(y)$ time.
As we argue below, if $|S_2'|<L/3$ then $|R_2'|$ is smaller than $(2c_2/3)\log (m+n)$, with high probability, and if $|S_2'|>L$ then $|R_2'| \geq (2c_2/3)\log (m+n)$ with high probability. These facts are easy consequences of Chernoff's bound; their proofs are given in Lemmas~\ref{lem:estimate2} and~\ref{lem:estimate}.

Thus, if $|R_2'|$ is smaller than $(2c_2/3)\log (m+n)$, and the number of pairs in $S_1$ is at most $L/3$, we stop the algorithm and proceed to Algorithm~\ref{sec:tec}.2.
Otherwise, we proceed to Stage III.

\paragraph{Stage III.}

We now handle the remaining case, where we assume that $|S_1| \geq L/3$ or $|S_2'|\geq L/3$, or both.

If $S_1$ contains at least $L/3$ pairs of points from $P\times Q$, we generate a
sample $R_1$ of $c_1\log (m+n)$ pairs of points from $S_1$ uniformly at random, for some sufficiently large constant $c_1>0$. (Otherwise, $R_1$ is taken to be empty.) As before, we generate this sample by picking a pair uniformly from a random bipartite clique $X\in \Gamma_L^1(P,Q, \alpha,\beta)$ where the probability of sampling $X$ is proportional to the number of pairs in $X$.

We assume that the distances between pairs of points in $P\times Q$ are distinct, and we use $R_1\cup R_2'$ to narrow $(\alpha,\beta]$. That is, we find two consecutive distances $\alpha',\beta'$ in $R_1\cup R_2'$ such that $\dsoneFrechet\in (\alpha',\beta']$,
using binary search with the decision procedure of Figure~\ref{alg:semi-sparse}.

As shown in Lemma~\ref{lem:median}, $R_1\cup R_2'$
contains, with high probability, an approximate median (in the \emph{middle three quarters}) of the distances between pairs in $S_1\cup S_2'$ --- the overall set of distances in $(\alpha,\beta]$.
This implies that $(\alpha',\beta']$ contains
at most $7/8$ of the distances in $(\alpha,\beta]$.

This terminates Stage III and the current iteration. We now
repeat these three stages with the narrowed interval $(\alpha',\beta']$.

\smallskip

Algorithm~\ref{sec:tec}.1 terminates when we meet the  termination criterion of Stage II.
 As will be argued, this happens, with high probability, when the current interval $(\alpha, \beta]$ contains at most $O(L)$ critical distances, including $\dsoneFrechet$, with a sufficiently small constant of proportionality.

We then proceed to Algorithm~\ref{sec:tec}.2 with the final narrowed interval.
We show in Lemma~\ref{lem:running_time1} that the resulting algorithm runs in
$O((m+n)^{4/3+\eps}/L^{1/3})$ time in expectation and with high probability, and uses $O((m+n)^{4/3+\eps}/L^{1/3})$ space, for any
$\eps>0$.

\paragraph{Remark.}
We note that our data structure is somewhat related to the data structure in~\cite{KS97},
but it is more suitable for our purpose. More specifically, the data structure of~\cite{KS97} uses several techniques, including the construction of an Eulerian path in an arrangement of disks, building a balanced segment tree, decomposition into smaller subproblems, dualization, and (one level of) $(1/r)$-cutting. $\Gamma_L(P,Q,\alpha,\beta)$ is somewhat simpler to construct as it only requires dualization and recursive $(1/r)$-cuttings. Computing the complete decomposition into $(1/r)$-cuttings, requires
$O((m+n)^{4/3+\eps})$ time and $O((m+n)^{4/3+\eps})$ storage, for any $\eps>0$, which is too expensive. However, our usage of recursion allows us to run the recursive decomposition in Stage I
until it reaches a level where the size of each subproblem is at most
$L$ and then stop, thereby making the algorithm more efficient.

\paragraph{Remark.}
Suppose that we have indeed narrowed down the interval $(\alpha, \beta]$, so
that it now contains $O(L)$ distances between pairs of $P\times Q$,
including $\dsoneFrechet$. We can then find $\dsoneFrechet$ by
simulating the execution of the decision procedure at the unknown
$\dsoneFrechet$. A simple way of doing this is as follows. To determine whether $M_{i,j}=1$  at
$\dsoneFrechet$, we compute the critical distance
$r' = |p_i-q_j|$ at which $M_{i,j}$ becomes $1$. If $r' \le \alpha$
then $M_{i,j}=0$, and if $r' \ge \beta$ then $M_{i,j}=1$.
Otherwise, $\alpha < r' < \beta$ is one of the $O(L)$ distances
in $(\alpha, \beta]$. In this case we run the decision procedure at
$r'$ to determine $M_{i,j}$. Since there are $O(L)$ distances
in $(\alpha, \beta]$, the total running time is $O(L(m+n))$. By
picking $L= (m+n)^{1/4+\eps}$ for another, but still arbitrarily
small $\eps > 0$, we balance
the bounds $O((m+n)^{4/3+\eps}/L^{1/3})$ and
$O(L(m+n))$, and obtain the bound $O((m+n)^{5/4+\eps})$, for any $\eps >0$,
on the overall running time.

Although this significantly improves the naive implementation
mentioned earlier, it suffers from the weakness that it has to run
the decision procedure separately for each distance in
$(\alpha, \beta]$ that we encounter during the simulation. In Algorithm~\ref{sec:tec}.2, described next,
we show how to accumulate several unknown distances and resolve them all
using a binary search that is guided by the decision procedure.
This allows us to find $\dsoneFrechet$ within the interval
$(\alpha, \beta]$ more efficiently.

\subsection{Algorithm~\ref{sec:tec}.2: An efficient simulation of the decision procedure}
\label{sec:searching_in_interval}
We assume that $\dsoneFrechet$ is in a given interval $(\alpha,\beta]$ that contains at most $L$ distances between pairs in $P\times Q$.
We simulate the decision procedure (of Figure~\ref{alg:semi-sparse}) at the unknown value $\delta^- = \dsoneFrechet$. The overall strategy of the simulation is to construct $S$ at $\delta^-$, one step at a time. Each such step
checks some specific entry $(i,j)$ of $M$ for being $0$ or $1$. To determine this, we need to compare
$\delta^-$ with some specific distance $r$ between a pair of points. At each step of the simulation we
will have a subrange $\tau$ of the original range $(\alpha,\beta]$, so that, for all values $\delta\in\tau$,
the simulation will make the same decisions, and therefore will construct a fixed prefix of the lowest
upward-skipping path $S$ up to the current location. To compare now $\delta^-$ with $r$,
we first test whether $r$ lies outside the current subrange $\tau$.
If $r$ lies to the left (resp., to the right) of $\tau$, we know that $\delta^- > r$ (resp., $\delta^- < r$), and can then
execute the current step of the simulation in a unique manner. If $r\in\tau$, we need to bifurcate, proceeding
along two separate branches, one assuming that $\delta^- < r$ (and then $M(i,j)=0$ at $\delta^-$)
and one assuming that $\delta^-\ge r$ (and then $M(i,j)=1$).
That is, the range $\tau$ of admissible values of $\delta$ is split by this bifurcation into the subranges
$\tau^- = \tau\cap(-\infty,r)$ and $\tau^+ = \tau\cap [r,\infty)$; we continue along one branch with $\tau^-$
and along the other with $\tau^+$.

Formally, these bifurcations generate a tree $T$. Each node $v$ of $T$ stores a pair $((i,j),\tau)$, where
$(i,j)$ is a location in $M$, and $\tau$ is a subrange of the initial range $(\alpha,\beta]$. The invariant
associated with this data is that, for each $\delta\in\tau$, the decision procedure at $\delta$ reaches the
location $(i,j)$ in the construction of the lowest upward-skipping sequence $S$, either as a $1$-entry
that becomes an element of $S$, or as a $0$-entry, which either lies immediately to the right of an
element of $S$, forcing $S$ to continue with an upward move, or as an entry inspected and skipped over
during an upward move. Another invariant that we maintain is that, for any subtree of $T$ rooted at the
root of $T$, the ranges stored at its leaves are pairwise disjoint and their union is
$(\alpha,\beta]$.

When we are at a node $v$ of $T$, we execute the next step of the construction of $S$, in which we
examine a suitable next entry $(i',j')$ of $M$. As explained above, this calls for comparing $\delta^-$
with the corresponding inter-point distance $r$. If we manage to resolve this comparison in a unique
manner, because $r$ happens to lie outside the current range $\tau$, we create a single child $v'$ of $v$,
and store in it the pair $((i',j'),\tau)$. Otherwise, we bifurcate. That is, we create two children $v^-$,
$v^+$ of $v$, and store at $v^-$ the pair $((i',j'),\tau^-)$, and store at $v^+$ the pair $((i',j'),\tau^+)$,
where $\tau^-$, $\tau^+$ are as defined above. It is clear that both invariants are maintained after either
of these (unary or binary) expansions.

To make the procedure efficient, we do not construct the entire tree $T$ at once, but proceed through a
sequence of phases. At each phase we start the construction from some node $\rho$, which represents some
unique prefix of the sequence $S$. We fix a threshold parameter $s$, whose value will be set later.
If, during the construction, we reach a node $v$ that has $s$ unary predecessors immediately preceding it,
we do not expand $T$ beyond $v$, and it becomes a leaf in the present version of $T$. We terminate the
phase either when every current leaf of $T$ has $s$ unary predecessors, as above, or when the total
number of nodes of $T$ becomes $m+n$.

We now sort the set $X$ of the $x=O(m+n)$ critical values at which
we have bifurcated (in $O((m+n) \log (m+n))$ time), we  then run a binary search for $\delta^-(P,Q)$ over $X$, using the decision procedure of Figure~\ref{alg:semi-sparse} to guide the search.
This step also takes $O((m+n)\log (m+n))$ time and determines all the $x$ unknown
values that we have encountered. Consequently we can identify the
path $\pi$ in $T$ that is the next portion of the overall lowest upward-skipping path $S$ in $M$
(at the optimal, still unknown, value $\delta^-$) which our decision procedure constructs. We also replace the current range $\tau$ of admissible values of $\delta$ by the subrange $\tau_v$ stored at the leaf of $\pi$.
We proceed in this manner through a sequence of such phases until, at the end of the simulation, we have a final range $\tau$ containing $\delta^-$, and we can then conclude that $\dsoneFrechet = \min\tau$ (it is easily checked that $\tau$ cannot be open at its left endpoint), and return this result as an output of the algorithm.

In Lemma~\ref{lem:searching_in_interval} we show that if we choose $s=(m+n)/L^{1/2}$, the (deterministic) running time of Algorithm~\ref{sec:tec}.2 is $O\left((m+n)L^{1/2}\log(m+n)\right)$. The space required by this algorithm is $O(m+n)$.

\subsection{The overall optimization algorithm.}\label{sec:overall_shortcuts}
We run Algorithm~\ref{sec:tec}.1 and then Algorithm~\ref{sec:tec}.2.

As noted earlier, Algorithm~\ref{sec:tec}.1 does not verify explicitly that the sample that it generates contains an
approximate median, nor does it verify that the number of distances in $(\alpha,\beta]$ is at most $O(L)$ when it terminates.
We prove, however, that these events occur with high probability.

As already mentioned, to balance the running times of Algorithms~\ref{sec:tec}.1 and \ref{sec:tec}.2,
we choose $L = (m + n)^{2/5+\eps}$,
for another, but still arbitrarily small $\eps > 0$. This gives the
following main result of this section.

\begin{theorem}
\label{thm:one_sided}
Given a set $P$ of $m$ points and a set $Q$ of $n$ points in the plane,
and a parameter $\eps>0$, we can compute the one-sided discrete Fr\'echet
distance $\dsoneFrechet$ with shortcuts in $O((m+n)^{6/5+\eps})$ time in expectation and with high probability using $O((m+n)^{6/5+\eps})$ space.
\end{theorem}

\paragraph{Remark.}
In principle, our algorithm for the one-sided discrete Fr\'echet distance with shortcuts can be generalized to higher dimensions. The only part that limits our approach to $\mathbb{R}^2$ is the use of cuttings in Algorithm~\ref{sec:tec}.1. However, this part can be replaced by a random sampling approach that is similar to the one that we use in Section~\ref{sec:semi_finding_interval} for the semi-continuous Fr\'echet distance with shortcuts. This will increase the running time of the algorithm, but it will stay strictly subquadratic. We omit here the details of this extension.

\subsection{Analysis and correctness}
\label{sec:correctness}

\begin{lemma}
\label{lem:cutting}
Given a set $P$ of $m$ points and a set
$Q$ of $n$ points in the plane, an interval
$(\alpha,\beta]\subset\reals$, and parameters $1\le L \le m+n$ and
$\eps>0$, Stage I of Algorithm~\ref{sec:tec}.1 constructs the data structure $\Gamma_L(P,Q,\alpha,\beta)= \Gamma_L^1(P,Q,\alpha,\beta) \cup \Gamma_L^2(P,Q,\alpha,\beta)$, in $O((m+n)^{4/3+\eps}/L^{1/3})$ time.
The sum of the sizes of the vertex sets of
the complete bipartite graphs of $\Gamma_L^1(P,Q,\alpha,\beta)$ and
 $\Gamma_L^2(P,Q,\alpha,\beta)$ is $O((m+n)^{4/3+\eps}/L^{1/3})$.
The total number
of pairs in $\Gamma_L^2(P,Q,\alpha,\beta)$ is
$O((m+n)^{4/3+\eps}L^{2/3})$.
\end{lemma}
\begin{proof}
Consider the recursion in Stage I of Algorithm~\ref{sec:tec}.1.
If we stop the recursion at
level $j$, we have $(m+n)/r^{3j} \approx L$, or
$r^j \approx ((m+n)/L)^{1/3}$. The number of subproblems is at most
$c^{2j}r^{4j} \approx c^{2j} ((m+n)/L)^{4/3}$. We choose $r=c^{2/\eps}$, so
we can bound $c^{2j}$ by ${(r^j)}^\eps$, where $\eps$ is the positive parameter
prespecified in the lemma.

The number of vertices of bicliques of $\Gamma_L^1(P,Q,\alpha,\beta)$ and $\Gamma_L^2(P,Q,\alpha,\beta)$ is dominated by the
size of the graphs output at the last level of the recursion, which is
$$
O(c^{2j} r^{4j} \cdot (m+n)/r^{3j}) =O((m+n)r^{j(1+\eps)})= O((m+n)^{4/3+\eps}/L^{1/3}).
$$

The total number
of pairs (whose distance is either in $(\alpha,\beta]$ or not) in the bipartite cliques of $\Gamma_L^2(P,Q,\alpha,\beta)$ is
$$O(((m+n)/L)^{4/3+\eps}\cdot L^2) = O((m+n)^{4/3+\eps}L^{2/3}),$$
since each subproblem at the bottom of the recursion contains at most $L^2$ edges.

The cost of constructing the structure is dominated by the cost of constructing the deepest $(1/r)$-cuttings, which is done one level before the last level of the recursion (i.e., at level $j-1$). In this level, we have $c^{2j-2}r^{4(j-1)}$ subproblems, each containing at most $(m+n)/r^{3(j-1)}$ points. As mentioned at the beginning of the description of Stage~I, constructing the primal $(1/r)$-cutting for such a subproblem costs $O(r\cdot (m+n)/r^{3(j-1)})$ time, and constructing $(1/r)$-cuttings of the duals of the $cr^2$ primal problems costs $O(cr^2\cdot r\cdot (m+n)/(r^{3(j-1)}r))$. Hence the overall cost of constructing the $(1/r)$-cuttings at this level is
$$
O(c^{2j-2} r^{4(j-1)}r^2 \cdot (m+n)/r^{3(j-1)}) = O((m+n)r^{j(1+\eps)}) = O((m+n)^{4/3+\eps}/L^{1/3}).
$$
\end{proof}


We now prove that the number of pairs of distance in $(\alpha,\beta]$ in the sample generated in Stage II of Algorithm~\ref{sec:tec}.1 is highly correlated with the number of pairs of $\Gamma_L^2(P,Q,\alpha,\beta)$ whose distance is in $(\alpha,\beta]$. To show that, we use the following multiplicative form of Chernoff's bound.

\begin{theorem}\textnormal{[Chernoff; see, e.g.,~\cite{bookChernoff}] }
\label{thm:chernoff}
Let $X_1,\ldots, X_R$ be independent random variables taking values in $\{0, 1\}$. Let $X=\sum_{i=1}^R X_i$ and let $\mu = E[X]$. Then, for any $\xi >0$ it holds that

(i) $Pr(X>(1+\xi)\mu)< \left(\dfrac{e^\xi}{(1+\xi)^{1+\xi}}\right)^\mu.$

Similarly, for any $0<\xi<1$, we have

(ii) $Pr(X<(1-\xi)\mu)< \left(\dfrac{e^{-\xi}}{(1-\xi)^{1-\xi}}\right)^\mu.$

\end{theorem}

Recall that $S_2$ is the set of pairs of $\Gamma_L^2(P,Q,\alpha,\beta)$, and $S_2'$ is the subset of pairs of $S_2$ whose distances are in $(\alpha,\beta]$. In Stage II of Algorithm~\ref{sec:tec}.1 we generated a random sample $R_2$ of $k=c_2(|S_2|/L)\log (m+n)$ pairs from $S_2$, and  $R_2'$ is the subset of pairs of $R_2$ whose distances are in $(\alpha,\beta]$. We now prove the following two lemmas.

\begin{lemma}
\label{lem:estimate2}
If $|S_2'|<L/3$ then the number of distances in $R_2'$ is smaller than $(2c_2/3)\log(m+n)$, with probability at least $1-\frac{1}{(m+n)^{c'}}$, for a constant $c'=\Theta(c_2)$.
\end{lemma}
\begin{proof}
For each pair $(p^i,q^i)\in R_2$, let $X_i$ be the indicator random variable of the event that $|p^i-q^i|\in (\alpha,\beta]$. Then, $X=\sum_{i=1}^k X_i=|R_2'|$ and $\mu = E(X) = |R_2|\cdot |S_2'|/|S_2|$.
Suppose that $|S_2'|$ is smaller than $L/3$. Then, $\mu< c_2(|S_2|/L)\log(m+n)\cdot (L/3)/|S_2| = (c_2/3)\log(m+n)$.

We fix a sufficiently large $\xi> 1$ so that $(1+\xi)\mu=(2c_2/3)\log(m+n)$.
Then, by Theorem~\ref{thm:chernoff}(i),
\begin{align*}
& Pr(|R_2'|> (2c_2/3)\log(m+n)) =Pr(X>(1+\xi)\mu)\\
& <\left(\dfrac{e^\xi}{(1+\xi)^{1+\xi}}\right)^\mu=\left(\dfrac{e^{\xi/(1+\xi)}}{(1+\xi)}\right)^{(1+\xi)\mu}\\ & <(e^{1/2}/2)^{(2c_2/3)\log(m+n)}\\ & <0.83^{(2c_2/3)\log(m+n)}\approx\frac{1}{(m+n)^{c'}},\end{align*}
for a suitable constant $c'$ that is proportional to $c_2$.
\end{proof}

\begin{lemma}
\label{lem:estimate}
If $|S_2'|>L$ then the number of distances in $R_2'$ is greater than $(2c_2/3)\log(m+n)$, with probability at least $1-\frac{1}{(m+n)^{c'}}$, for some constant $c'=\Theta(c_2)$. \end{lemma}
\begin{proof}
As in Lemma~\ref{lem:estimate2}, for each pair $(p^i,q^i)\in R_2$, let $X_i$ be the indicator random variable of the event that $|p^i-q^i|\in (\alpha,\beta]$. Then, as before, $X=\sum_{i=1}^k X_i=|R_2'|$ and $\mu = E(X) = |R_2|\cdot |S_2'|/|S_2|$.  Suppose that $|S_2'|$ is greater than $L$. Then, $\mu > c_2(|S_2|/L)\log(m+n)\cdot L/|S_2| = c_2\log(m+n)$.

Set $\xi=1/3$. Then, by Theorem~\ref{thm:chernoff}(ii),
\begin{align*}
& Pr(|R_2'|< (2c_2/3)\log(m+n))< Pr(X <(1-\xi)\mu)\\
& < \left(\dfrac{e^{-\xi}}{(1-\xi)^{(1-\xi)}}\right)^\mu<\left(\dfrac{e^{-1/3}}{(2/3)^{2/3}}\right)^{c_2\log(m+n)}\\ & < (0.94)^{c_2\log(m+n)} \approx\frac{1}{(m+n)^{c'}},\end{align*}
for a suitable constant $c'$ that is proportional to $c_2$.
\end{proof}

\begin{lemma}
\label{lem:median} If $|S_1|\geq L/3$ or $|S_2'|\geq L/3$ (or both),
then, with
high probability, the sample $R_1\cup R_2'$
contains a pair whose distance is in the middle three quarters of the sequence of sorted
distances between pairs of points of $P\times Q$ that lie in $(\alpha,\beta]$.
\end{lemma}

\begin{proof}
If $|S_1| \ge L/3$ then the sample $R_1$
contains, with high probability, a pair whose distance is in the middle half of the sequence of sorted distances between pairs in $S_1$.
Indeed, the probability that $R_1$ does not contain a pair $(p',q')$ of points whose distance is in the middle half of the distances recorded in $\Gamma^1_L(P,Q,\alpha,\beta)$, is $(1/2)^{c_1\log (m+n)} = 1/(m+n)^{c_1}$. 

If $|S'_2| \ge L/3$ then the probability that $R_2$ does not contain a pair $(p'',q'')$ at distance in the middle half of the sequence of sorted distances between pairs of $S_2'$ is
$$
\left(1-\dfrac{|S_2'|}{2|S_2|}\right)^{c_2(|S_2|/L)\log(m+n)} < e^{-\frac{1}{2}c_2(|S_2'|/L)\log(m+n)}\leq e^{-\frac{1}{6}c_2\log(m+n)} =
\dfrac{1}{(m+n)^{c'}},
$$

for a suitable constant $c'$ that is proportional to $c_2$.
Let $S$ be the larger among $S_1$ and $S_2'$, and let $R$ be the corresponding sample ($R_1$ or $R_2'$) from $S$.
Then, $|S|\geq L/3$, and, with high probability, $R$ contains a pair $(p',q')$ whose distance,
$d'$, lies in the middle half of the pairwise distances of $S$. Thus, at least $1/4$ of the distances in $S$ are smaller than $d'$, and so $d'$ is greater than at least $1/8$ of the distances in $S_1\cup S_2'$. By a similar reasoning $d'$ is also smaller than at least $1/8$ of the distances in $S_1\cup S_2'$. Thus, it is in the
middle three quarters of $S_1\cup S_2'$ --- the overall set of pairs whose distances are in $(\alpha,\beta]$.
\end{proof}



\paragraph{Correctness.}
Follows immediately from the fact that the interval $(\alpha,\beta]$ contains $\delta^-(P,Q)$ with certainty, and from the correctness of the decision procedure.

We next bound the running time of the algorithm.

\begin{lemma}
\label{lem:running_time1}
Algorithm~\ref{sec:tec}.1 runs in $O((m+n)^{4/3+\eps}/L^{1/3})$ time in expectation and with high probability, and uses $O((m+n)^{4/3+\eps}/L^{1/3})$ space, for any
$\eps>0$. Furthermore, when Algorithm~\ref{sec:tec}.1 terminates
$(\alpha,\beta]$ contains less than $4L/3$ distances with high probability.
\end{lemma}
\begin{proof}
Constructing $\Gamma_L(P,Q,\alpha,\beta)$ in Stage I takes $O((m+n)^{4/3+\eps}/L^{1/3})$ time, as shown in Lemma~\ref{lem:cutting}.

The time for generating the samples in Stage II and Stage III is (at worst) proportional to the size of $\Gamma_L^1 (P,Q,\alpha,\beta)$ and $\Gamma_L^2(P,Q,\alpha,\beta)$, which by Lemma~\ref{lem:cutting} is $O((m+n)^{4/3+\eps}/L^{1/3})$.

The number of sampled pairs in $R_1$ is $O(\log(m+n))$, and the number of pairs in $R_2'$ is at most $|R_2|$, which is
$$O((|S_2|/L)\log (m+n)) = O\left(\left(((m+n)^{4/3+\eps}L^{2/3})/ L\right)\log (m+n)\right) =  O\left((m+n)^{4/3+\eps}/L^{1/3}\right),$$
for an arbitrarily small $\eps>0$ (the logarithmic factor is absorbed in the last bound by slightly increasing $\eps$).
Thus, the running time for finding two consecutive distances in $R_1\cup R_2'$ in Stage III (that delimit the narrowed interval $(\alpha',\beta']$) is subsumed in the bounds on the costs of Stages I and II.

By Lemma \ref{lem:estimate2} when $(\alpha,\beta]$
contains less than $2L/3$ distances (then it must have either less than $L/3$ distances in $S_2'$ or less than
$L/3$ distances in $S_1$) then with high probability we stop at Stage II.
Furthermore, by
Lemma~\ref{lem:median}, if $(\alpha,\beta]$ contains
at least $2L/3$ distances then with high probability, we narrow $(\alpha, \beta]$  to an interval $(\alpha', \beta']$ that contains at most $7/8$ of the distances that were in $(\alpha,\beta]$. It follows that we repeat the three stages only $O(\log(m+n))$ times with high probability.

We conclude that the resulting algorithm runs in
$O((m+n)^{4/3+\eps}/L^{1/3})$ time with high probability, for any
$\eps>0$ (again we absorb an additional logarithmic factor by slightly increasing $\eps$, as above).
It may happen with polynomially small probability that the length of the new interval $(\alpha',\beta]$ is larger than $7/8$ times the length of $(\alpha,\beta]$, or that we do not stop when
$(\alpha,\beta]$
contains less than $2L/3$ distances.
But even when these rare events happen the algorithm still runs in polynomial time. It follows that the expected running time of the algorithm is also $O((m+n)^{4/3+\eps}/L^{1/3})$.

By
Lemma \ref{lem:estimate} if $(\alpha,\beta]$ contains more than $L+L/3$ distances (then it must have either more than $L$ distances in $S_2'$ or more than
$L/3$ distances in $S_1$) then with high probability
Algorithm~\ref{sec:tec}.1 does not stop at Stage II.
This implies that with high probability when
Algorithm~\ref{sec:tec}.1 stops $(\alpha,\beta]$ contains at most $4L/3$
distances.

The space required by our algorithm is proportional to the size of $\Gamma_L^1(P,Q,\alpha,\beta)\cup\Gamma_L^2(P,Q,\alpha,\beta)$, and thus it is $O((m+n)^{4/3+\eps}/L^{1/3})$ (again, according to Lemma~\ref{lem:cutting}).
\end{proof}

We now analyze the running time and the storage requirement of Algorithm~\ref{sec:tec}.2.
\begin{lemma}
\label{lem:searching_in_interval}
Given a set $P$ of $m$ points and a set $Q$ of $n$ points in the plane, and
an interval $(\alpha,\beta]\subset\reals$ that contains $O(L)$ distances
between pairs in $P\times Q$, including $\dsoneFrechet$, Algorithm~\ref{sec:tec}.2 finds
$\dsoneFrechet$ deterministically in $O((m+n)L^{1/2}\log(m+n))$ time using $O(m+n)$ space.
\end{lemma}

\begin{proof}
The simulation of the decision procedure in Algorithm~\ref{sec:tec}.2 is partitioned into phases, where in each phase we generate a tree $T_d$, consisting of the bifurcations made during the traversal of $M$ and the paths connecting these bifurcations, and resolve the comparisons associated with it. There are two kinds of phases. A phase is called \emph{successful} if it extends the desired upward-skipping path $S$ by at least $s$ steps, and it is called \emph{unsuccessful} otherwise.
 Since the total size of $S$ is $O(m+n)$ there can be at most $O((m+n)/s)$ successful phases of this kind, whose total cost is thus $O(((m+n)^2/s)\log (m+n))$ (recall that resolving the comparisons in a single phase requires a logarithmic number of calls to the decision procedure).

It remains to analyze the cost of unsuccessful phases. The size of $T_d$ in an unsuccessful phase must be $m+n$, as otherwise any path from the root to a leaf in $T_d$ is of length at least $s$ and the phase must have been successful. On the other hand, if we denote by $x$ the number of bifurcating nodes (for simplicity assume the root of $T_d$ is bifurcating) in $T_d$, then the size of $T_d$ is $O(xs)$. Indeed, each node which is not bifurcating is on at least one path of length at most $s$ emanating from a bifurcating node. Since two such paths leave each bifurcating node, their total number is $2x$ and their total size is at most $2xs+x$. It follows that $m+n\leq 2xs+x\leq 3xs$ or $x\geq \frac{m+n}{3s}$. Note that $x$ is the number of critical values that we encounter in this phase, and that, by construction, each critical value can arise in a single phase. Hence,
the number of unsuccessful phases is $O(Ls/(m+n))$. Each such phase takes $O((m+n)\log (m+n))$
time, as before, for a total of $O(Ls\log(m+n))$ time.

Overall, the cost is thus
$$
O\left( \dfrac{(m+n)^2\log (m+n)}{s} + Ls\log(m+n) \right) ,
$$
which, by choosing $s=(m+n)/L^{1/2}$, becomes $O\left((m+n)L^{1/2}\log(m+n)\right)$.

Note that, after each phase  of the algorithm, we can free the memory
used to process the phase and only remember $\alpha,\beta$ and the
path in $M$ (a prefix of the desired $S$) that we have traversed so
far. Since each phase processes $O(m+n)$ entries of $M$, the space
needed by this algorithm is $O(m+n)$.
\end{proof}

Theorem \ref{thm:one_sided} follows from Lemma \ref{lem:running_time1} and Lemma \ref{lem:searching_in_interval}.

\section{Approximating the $k$th distance (by rank)}
\label{sec:cor:aprox_selection}
In this section, we step out of the context of the Fr\'echet distance.
As already noted, we believe that Algorithm~\ref{sec:tec}.1 (of Section~\ref{sec:finding_interval}) is of independent interest,
and that it may find other applications for distance-related optimization
problems. In this section we give one such application for approximate distance selection.

Given a set $A$ of $m$ points and a set $B$ of $n$ points in the plane,
one can find a pair $(a,b) \in A\times B$ such that $|a-b|$ is the
$k$-th smallest distance between a point of $A$ and a point of $B$,
using an algorithm of Katz and Sharir~\cite{KS97},
that runs in
$O((m+n)^{4/3}\log^2 (n+m))$ time. In fact, using just the first part of their algorithm,
we can decide, for a given threshold distance $\delta$, whether the number $N$ of pairs in $A\times B$ at distance
at most $\delta$ is at most $k$, in $O((m+n)^{4/3}\log (n+m))$ time.
The following theorem shows that if we do not insist on obtaining the
pair realizing the $k$-th smallest distance exactly, but are willing
to get by with a pair realizing a distance of rank sufficiently close to $k$,
we can speed up the computation.

\begin{theorem}
\label{cor:aprox_selection}
Let $A$  be a set of $m$ points and let $B$ be a set of $n$ points in the plane,
and let $k, L$, and $t$ be such that $0<k<mn$, $\sqrt{k}\le L\le k$, and $t=L^2/k$.

\noindent (a)
We can find a pair
$(a,b) \in A\times B$ such that, with high probability, $|a-b|$ is the
$\kappa$-th smallest distance between a point of $A$ and a point of $B$,
for some rank $\kappa$ satisfying $k-L\le \kappa \le k+L$, in worst-case
${\displaystyle O\left(\frac{mn}{t}\log (m+n) + m +n \right) =
O\left(\frac{mnk}{L^2}\log (m+n) + m +n\right)}$ time and storage.

\noindent (b)
If $t\le m+n$ we can also find such a pair in
$O((m+n)^{4/3+\eps}/t^{1/3})=O((m+n)^{4/3+\eps}k^{1/3}/L^{2/3})$ time and space for any $\eps>0$. This time bound holds in expectation and with high probability.
\end{theorem}

\noindent
{\bf Remark.}
Note that the bound in part (b) of the theorem is better than that in (a) when
(in what follows, the polylogarithmic factors are suppressed, as they can be
subsumed in the stated expressions, by slightly increasing the value of $\eps$)
\begin{align} \label{tthres}
\frac{L^2}{k} = t & < \frac{(mn)^{3/2}}{(m+n)^{2+\frac32 \eps}} ,\quad\text{or} \\
L & < \frac{k^{1/2}(mn)^{3/4}}{(m+n)^{1+\frac34 \eps}} \nonumber .
\end{align}
Since the right-hand side of the first inequality of (\ref{tthres}) is always smaller than $m+n$,
this is the effective threshold where (b) should be used, instead of
(the simpler procedure in) (a).

\begin{proof}

We first consider the situation in part (b), and assume that $t\le m+n$.
Consider first the case where $L\le m+n$ too.
We use the following randomized decision procedure. Given a distance parameter $\delta$, let
$N$ denote the number of pairs in $A\times B$ at distance at most $\delta$.
Let the parameters $k$ and $L$ be as specified in the theorem. Given $\delta$, $k$, and $L$,
the decision procedure returns ``SMALL'' if $N < k-L$, ``LARGE'' if  $N> k+L$, and, in
case $k-L \le N \le k+L$, it may return either SMALL or LARGE. This output
is guaranteed with high probability (i.e., the procedure errs with probability polynomially small
in $m+n$). The (worst-case) running time\footnote{%
  In contrast with the \emph{deterministic} decision procedure of Section \ref{sec:finding_interval},
  the procedure here is randomized and may err with small probability. Thus our algorithm
  may not give a correct answer, but this happens with polynomially small probability.}
is, as shown below, $O((m+n)^{4/3+\eps}/t^{1/3}) = O((m+n)^{4/3+\eps}k^{1/3}/L^{2/3})$.

We first prove the theorem, assuming the availability of such a decision procedure,
and then describe the decision procedure itself.

We find an interval
$[\alpha, \beta]$ that contains (with high probability) at most $L$ distances between pairs
of points from $A\times B$, so that at least one of these distances is of rank in $[k-L,k+L]$.
We then return either $\alpha$ or $\beta$ (together with its generating pair).
Clearly this distance is of rank in $[k-2L,k+2L]$. Rescaling $L$ by half, we get the
desired procedure.

To find an interval $[\alpha, \beta]$ with these properties, we start
with $[\alpha,\beta]=[0,\infty)$ and repetitively shrink it while maintaining the
property that it contains a distance of rank in $[k-L,k+L]$. We stop when we determine
(with high probability) that it contains at most $L$ distances between pairs of points from $A\times B$.

At each step of the search, we construct the hierarchical tree-like cuttings,
as in the algorithm of Section~\ref{sec:finding_interval}, where the construction stops
at the level where the size of each subproblem is at most $L$.
If the algorithm reports that the number of critical distances in
$[\alpha,\beta]$ is at most\footnote{%
  Note that if $L>m+n$, the hierarchy of cuttings is empty, and the whole graph
  $A\times B$ is left undecomposed. This situation will be handled later.}
$L$, we stop. Otherwise, the algorithm produces a random
sample $R$ of the distances in $[\alpha,\beta]$, that contains, with high probability,
an approximate median (in the middle three quarters) of the pairwise distances in $[\alpha,\beta]$.

Let $d_1$ (resp., $d_2$) be the smallest (resp., largest) distance of a pair in $R$.
If the decision procedure returns LARGE for $d_1$, we know that the rank of $d_1$ is
$\ge k-L$, and we shrink $[\alpha,\beta]$ to $[\alpha,d_1]$, noting that our invariant
is maintained. Indeed, $[\alpha,\beta]$ contains a distance $d$ of rank in $[k-L,k+L]$.
If $\alpha\le d\le d_1$, we are done; otherwise, the rank of $d_1$ must be in $[k-L,k+L]$
and the invariant is again maintained.
If the procedure returns SMALL for $d_2$, we shrink $[\alpha,\beta]$ to $[d_2,\beta]$,
and an argument symmetric to the one just given shows that the invariant is maintained
in this case too.  Otherwise, there exists at least one
consecutive pair $x < y$ of distances in $R$,
such that the decision procedure returns SMALL for $x$ and LARGE for $y$.
We locate one such pair using binary search, and shrink $[\alpha,\beta]$ to $[x,y]$.
Again, since we know that the rank of $x$ is $\le k+L$ and that of $y$ is $\ge k-L$,
it is easily verified that $[x,y]$ contains a distance of rank in $[k-L,k+L]$.
Each of the preceding arguments holds with high probability.
Finally, since $R$ contains, with high probability, an approximate median (in the
middle three quarters of the distances in $[\alpha,\beta]$), it follows that
the number of distances in $[x,y]$ is, with high probability,
at most ${7/8}$ of the number of distances in $[\alpha,\beta]$. This argument applies also
to the extreme cases, when we shrink the interval to $[\alpha,d_1]$ or to $[d_2,\beta]$.
We then proceed to the next step of the search with the shrunk interval.

Since each call to the decision procedure errs with polynomially small probability,
it follows that, with high probability, we return, upon termination, a pair $(a,b) \in A\times B$
whose distance is of rank in $[k-2L, k+2L]$.
The running time of the overall resulting algorithm is dominated, up to a polylogarithmic factor,
by the cost of the decision procedure, so, with an appropriate adjustment of $\eps$,
it runs in worst-case randomized $O((m+n)^{4/3+\eps}k^{1/3}/L^{2/3})$ time,
uses $O((m+n)^{4/3+\eps}k^{1/3}/L^{2/3})$ space, and returns a correct output with high probability.


\paragraph{The decision procedure.}
We replace the annuli centered at the points of $A$ and $B$, as used by the original algorithm,
by respective disks of radius $\delta$ centered at the same points, and compute a hierarchical
cutting, as in Section~\ref{sec:finding_interval}, obtaining a collection of complete bipartite graphs of
disks and points, so that within each subgraph, all the points are contained in all the disks.
We compute this hierarchical cutting until we reach subproblems of size $t:=L^2/k$
instead of the size $L$ used originally and in the selection procedure described above.
recall that, by assumption, we have $1\le t\le L < m+n$. The
procedure ends up with sets $S_1$, $S_2$ of pairs of $A\times B$, as before,
where all the pairs in $S_1$ are at distances at most $\delta$, while only
some of the pairs in $S_2$ have this property; we let $S'_2$ denote the subset
of pairs in $S_2$ at distance $\le\delta$. We estimate $|S'_2|$ by drawing a
random sample $R_2$ from $S_2$, consisting of $\frac{c_2|S_2|}{t}\log(m+n)$ pairs,
for a suitable constant $c_2$, and by explicitly counting the number of pairs in
$R'_2:=R_2\cap S'_2$. As already stated, we want to detect the
cases $|S_1|+|S'_2| < k-L$ and $|S_1|+|S'_2| > k+L$, with high probability.

Note that we know the exact value of $|S_1|$, so, putting $k_0:=k-|S_1|$,
we want to detect the cases $|S'_2| < k_0-L$ and $|S'_2| > k_0+L$.
To do so, we compute $|R'_2|$, and report that $|S'_2| < k_0-L$ if
${\displaystyle |R'_2| \le \frac{c_2 k_0}{t} \log(m+n)}$, and that
$|S'_2| > k_0+L$ if ${\displaystyle |R'_2| \ge \frac{c_2 k_0}{t} \log(m+n)}$.

To show that the error probability of either decision is small, we use the following lemma, which
replaces Lemmas~\ref{lem:estimate2} and \ref{lem:estimate}.
(Note that part (a) of the lemma becomes vacuous when $k_0\le L$.)
\begin{lemma} \label{cher3.2}
(a) If $|S'_2| < k_0-L$ then
${\displaystyle |R'_2| \le \frac{c_2 k_0}{t} \log(m+n)}$ with
probability at least $1-\frac{1}{(m+n)^{c'}}$, for $c' = \Theta(c_2)$. \\
(b) If $|S'_2| > k_0+L$ then
${\displaystyle |R'_2| \ge \frac{c_2 k_0}{t} \log(m+n)}$ with
probability at least $1-\frac{1}{(m+n)^{c'}}$, for $c' = \Theta(c_2)$.
\end{lemma}
\begin{proof}
(a) Arguing as in the previous proof (and using the same terminology), we now have
\begin{align*}
\mu & = E(X) = |R_2|\cdot |S'_2|/|S_2| \le
\frac{c_2 |S_2|}{t}\log(m+n) \cdot \frac{k_0-L}{|S_2|} =
c_2 \frac{k_0-L}{t} \log(m+n).
\end{align*}
We define $\xi$ such that the following equation holds
$$
(1+\xi)\mu = c_2 \frac{k_0}{t} \log(m+n) \ .
$$
So it follows that
$$
\xi\mu = c_2 \frac{k_0}{t} \log(m+n) - \mu \ge c_2 \frac{L}{t} \log(m+n) .
$$
Applying Theorem~\ref{thm:chernoff}(i), the probability that
${\displaystyle |R'_2| > (1+\xi)\mu=\frac{c_2 k_0}{t} \log(m+n)}$ is at most
$$
\left( \frac{e^\xi}{(1+\xi)^{1+\xi}}\right)^\mu < e^{-\xi^2\mu/3}
= e^{-(\xi\mu)^2/(3\mu)} \le e^{-\frac{c_2}{3}\cdot\frac{L^2\log(m+n)}{t(k_0-L)}}
\le e^{-\frac13 c_2\log(m+n)} = \frac{1}{(m+n)^{c'}} ,
$$
where the last inequality follows from the choice of $t$, and $c'$ is a constant proportional to $c_2$.

\smallskip
\noindent
(b) Here we have
\begin{align*}
\mu & = E(X) = |R_2|\cdot |S'_2|/|S_2| \ge
\frac{c_2 |S_2|}{t}\log(m+n) \cdot \frac{k_0+L}{|S_2|} =
c_2 \frac{k_0+L}{t} \log(m+n).
\end{align*}
We pick $\xi$ such that
$$
(1-\xi)\mu = c_2 \frac{k_0}{t} \log(m+n)\ .
$$
So it follows that
\begin{align*}
\xi\mu & = \mu - c_2 \frac{k_0}{t} \log(m+n) ,\quad\text{or} \\
\xi & = 1 -  \frac{c_2 k_0}{\mu t} \log(m+n) .
\end{align*}
Now, applying Theorem~\ref{thm:chernoff}(ii), the probability that
${\displaystyle |R'_2| < \frac{c_2 k_0}{t} \log(m+n)}$ is at most
$$
\left( \frac{e^{-\xi}}{(1-\xi)^{1-\xi}}\right)^\mu < e^{-\xi^2\mu/2} .
$$
The expression in the exponent satisfies
$$
\xi^2\mu = \left( 1 -  \frac{c_2 k_0}{\mu t} \log(m+n) \right)^2 \mu ,
$$
which is clearly an increasing function of $\mu$. Hence,
substituting the minimum possible value of $\mu$, we have
$$
\xi^2\mu \ge \left( 1 -  \frac{k_0}{k_0+L} \right)^2
c_2 \frac{k_0+L}{t} \log(m+n) = \frac{c_2L^2}{(k_0+L)t} \log(m+n) \ge \frac{c_2}{2} \log(m+n) ,
$$
by the choice of $t$. Hence in this case the failure probability is at most
$$
\frac{1}{(m+n)^{c'}} ,
$$
where $c'$ is a constant proportional to $c_2$, as claimed.
\end{proof}


The running time of the decision procedure
is
$$
O\left((m+n)^{4/3+\eps}/t^{1/3}  \right) =
O\left((m+n)^{4/3+\eps} k^{1/3}/L^{2/3}  \right) ,
$$
for any $\eps>0$.
This follows by an analysis as in the proof of Lemma \ref{lem:running_time1}.

To complete the analysis of case (b), assume next that $L \ge m+n$ but $t$
is still at most $m+n$.
We replace the hierarchical construction for shrinking the interval $[\alpha,\beta]$
by the following simpler approach. We sample a set $S$ of  $O((m+n)\log(m+n))$ pairs
from $A\times B$. Standard probabilistic reasoning shows that, with high probability,
the maximum number of unsampled pairs between any pair of consecutive elements of $S$
(in the order sorted by distance) is $O\left( \frac{mn}{m+n} \right)$, which can be
assumed to be at most $L$, with a suitable choice of the constant of proportionality.
Hence, with high probability, for each interval of $L$ consecutive pairs of $A\times B$
(in the distance-sorted order), $S$ contains at least
$\Omega\left( \frac{L(m+n)}{mn} \right) = \Omega(1)$ pairs from the interval. In the following paragraph we assume that this event does occur.

We now use binary search to find
a pair $(p,q)$ such that the rank of $d(p,q)$ is in $[k-2L,k+2L]$ with high probability, as follows.
Let $d_{\rm min}$ be the smallest distance of a pair in the sample.
If the decision procedure returns LARGE for $d_{\rm min}$ then the rank of $d_{\rm min}$ is at least $k-L$ with high probability, and it is
at most $L\le k$, as we know that there is a pair in the sample among every $L$ consecutive pairs in the distance-sorted order  of all pairs.
Similarly,
let $d_{\rm max}$ be the largest distance of a pair in the sample.
If the decision procedure returns SMALL for $d_{\rm max}$ then the rank of $d_{\rm max}$ is at most $k+L$ with high probability and, for the same reason as above, it is
at least $k-L$. In these cases we return (the pair realizing) $d_{\rm min}$ or $d_{\rm max}$ as the desired output.
If the decision procedure returns SMALL for $d_{\rm min}$ and LARGE for $d_{\rm max}$
we apply binary search, using our  decision procedure, to find
 two pairs $(p_1,q_1)$ and $(p_2,q_2)$ in $S$,
which are consecutive in the distance-sorted order of the pairs in $S$, such that
the decision procedure returns SMALL for $(p_1,q_1)$ and LARGE for $(p_2,q_2)$.
Since there are at most $L$ pairs in between $(p_1,q_1)$ and  $(p_2,q_2)$ in the distance sorted order of all pairs,
the ranks
 of both $d(p_1,q_1)$ and $d(p_2,q_2)$ must then be in
the range $[k-2L,k+2L]$, so we can return either $x$ or $y$
as the desired output.

The running time is still dominated by the running time of the decision procedure
(times a logarithmic factor, which can be ignored by slightly increasing $\eps$), which is
$$
O\left((m+n)^{4/3+\eps}/t^{1/3} \right) =
O\left((m+n)^{4/3+\eps} k^{1/3}/L^{2/3} \right) ,
$$
for any $\eps>0$.

Finally, consider case (a). Here we make no assumptions concerning $t$ and $L$. We
replace the decision procedure by the following simpler one.
We  sample a subset $R$ of
$$
\frac{cmn}{t}\log (m+n)
$$
pairs from $A\times B$, for a suitable absolute constant $c$.
We then set $$i:= \frac{ck}{t}\log (m+n),$$ and find (in $O(|R|)$ time) the $i$-th
element $d_i=(a_i,b_i)$ of $R$, in the distance-sorted order, and report it as the output pair.
The following lemma establishes the correctness of this decision procedure.
\begin{lemma}
The rank of the pair $d_i=(a_i,b_i)$  in the distance-sorted order of all pairs in $A\times B$ is in $[k-L,k+L]$ with high probability.
\end{lemma}
\begin{proof}
The claim follows by showing that,
 with high probability, we sample at least $i$ pairs of rank no larger than $k+L$ and at most $i$ pairs of rank at
least $k-L$. We prove these properties under the sampling model in which
each pair is sampled independently with probability
$\frac{c}{t}\log (m+n)$. The alternative model, in which each subset of size $\frac{cmn}{t}\log (m+n)$ is sampled with equal probability, can be handled in a similar, but slightly more complicated, manner.

Consider the $k+L$ pairs of smallest rank.
Let $X_1$ be a random variable equal to the number of such pairs that are contained in $R$. Since each pair is sampled with probability
$\frac{c}{t}\log (m+n)$, we have
$\mu_1 := E(X_1) = \frac{c(k+L)}{t}\log (m+n)$.

Let $\xi = \frac{L}{k+L}$; so we have
$$
(1-\xi)\mu_1 = i = \frac{ck}{t}\log (m+n) \ .
$$

Applying Theorem \ref{thm:chernoff}(ii), we have that the probability that $X < i$ is at most
$$
\left( \frac{e^{-\xi}}{(1-\xi)^{1-\xi}}\right)^{\mu_1} < e^{-\xi^2\mu_1/2} = e^{-\left(\frac{L}{k+L}\right)^2\frac{c(k+L)}{2t}\log (m+n)} =
e^{-\frac{ck}{2(k+L)}\log (m+n)}.
$$
Since $L \le k$ this is smaller than $1/(m+n)^{c'}$ for some constant $c'$ that is proportional to $c$.

Similarly, consider now the $k-L$ pairs of smallest rank.
Let $X_2$ be a random variable equal to the number of such pairs that are contained in $R$. Again, since each pair is sampled with probability
$\frac{c}{t}\log (m+n)$, we have
$\mu_2 := E(X_2) = \frac{c(k-L)}{t}\log (m+n)$.

Let $\xi = \frac{L}{k-L}$; we have
$$
(1+\xi)\mu_2 = i = \frac{ck}{t}\log (m+n) \ .
$$

Applying Theorem \ref{thm:chernoff}(i), we have that the probability that $X > i$ is at most
$$
\left( \frac{e^{\xi}}{(1+\xi)^{1+\xi}}\right)^{\mu_2} < e^{-\xi^2\mu_2/3} = e^{-\left(\frac{L}{k-L}\right)^2\frac{c(k-L)}{3t}\log (m+n)} =
e^{-\frac{ck}{3(k-L)}\log (m+n)}.
$$
Since $L \le k$  this is also smaller than $1/(m+n)^{c'}$ for some constant $c'$ that is proportional to $c$.

\end{proof}

The cost of this procedure is $O(|R|+m+n) = O\left( \frac{mn}{t}\log (m+n) +m+n \right)$.
(The linear terms are added since in any case we need to read the input.)
This establishes part (a), and thereby completes the proof of the theorem.
\end{proof}

\noindent
{\bf Remark.}
We note that when $L < \sqrt{k}$ we have $t<1$, so the algorithm, as presented, does not
apply, and the best we can do is to run the exact selection procedure, which takes
$O((m+n)^{4/3}\log^2 (m+n))$ time.


\section{An efficient algorithm for computing the discrete Fr\'echet distance with two-sided shortcuts}
\label{sec:two_sided}

We first consider the corresponding decision problem.
That is, given $\delta>0$, we wish to decide whether
$\dstwoFrechet\le\delta$.\footnote{We ignore the issue of
discrimination between the cases of strict inequality and equality,
which will be handled in the optimization procedure, described later.}

Consider the matrix $M$ as defined in the Section~\ref{sec:DFDS1}.
In the two-sided version of DFDS, given a reachable position
$(p_{i},q_{j})$ of the frogs, the $P$-frog can make a \emph{skipping
upward move}, as in the one-sided variant, to any point $p_{k},k>i$,
for which $M_{k,j}=1$. Alternatively, the $Q$-frog can jump to any
point $q_{l},l>j$, for which $M_{i,l}=1$; this is a \emph{skipping
right move} in $M$ from $M_{i,j}=1$ to $M_{i,l}=1$. Determining whether $\dstwoFrechet\le\delta$
corresponds to deciding whether there exists a \emph{skipping
row- and column-monotone path} of ones in $M$ that starts at $(0,0)$, ends at
$(m-1,n-1)$, and consists of an interweaving sequence of skipping
upward moves and skipping right moves; see
Figure~\ref{fig:staircase}(c)).

Katz and Sharir~\cite{KS97} showed that the set
$S=\{(p_{i},q_{j})\mid\|p_{i}-q_{j}\|\le\delta\}=\{(p_{i},q_{j})\mid
M_{i,j}=1\}$ can be computed, in $O((m^{2/3}n^{2/3}+m+n)\log n)$
time and space, as the union of the edge sets of a collection
$\Gamma=\{P_{t}\times Q_{t} \mid P_{t}\subseteq P, \; Q_{t}\subseteq
Q\}$ of edge-disjoint complete bipartite graphs. The number of
graphs in $\Gamma$ is $O(m^{2/3}n^{2/3}+m+n)$, and the overall sizes
of their vertex sets are
\[
\sum_{t}|P_{t}|,\sum_{t}|Q_{t}|=O((m^{2/3}n^{2/3}+m+n)\log n).
\]
We store each graph
$P_{t}\times Q_{t}\in\Gamma$ as a pair of sorted linked lists $L_{P_{t}}$
and $L_{Q_{t}}$ over the points of $P_{t}$ and of $Q_{t}$, respectively.
For each graph $P_{t}\times Q_{t}\in\Gamma$, there is $1$ in each entry $(i,j)$ such that $(p_{i},q_{j})\in P_{t}\times Q_{t}$.
That is, $P_t \times Q_t$ corresponds to a submatrix
$M^{(t)}$ of ones in $M$ (whose rows and columns are not necessarily
consecutive). See Figure~\ref{fig:submatrices}(a).

Note that if $(p_{i},q_{j})\in P_{t}\times Q_{t}$ is a reachable
position of the frogs, then every pair in the set $\{(p_{k},q_{l})\in P_{t}\times Q_{t}\mid k\ge i,l\ge j\}$
is also a reachable position. (In other words, the positions
in the upper-right submatrix of $M^{(t)}$ whose lower-left entry
is $(i,j)$ are all reachable; see Figure~\ref{fig:submatrices}(b)).

\begin{figure}[htb]
\centering \includegraphics[scale=0.6]{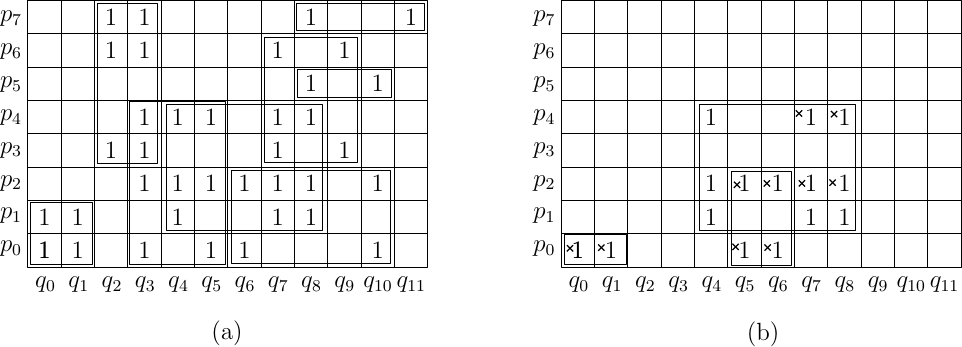}
\centering \caption{\small(a) A possible representation of the matrix $M$ as a collection of submatrices of ones, corresponding to the complete bipartite graphs
$\{p_0,p_1\}\times\{q_0, q_1\},
 \{p_0, p_2, p_4\}\times \{q_3, q_5\},
 \{p_0,p_2\}\times\{q_6, q_{10}\},
 \{p_1, p_2,p_4\}\times\{q_4, q_7, q_8\},
 \{p_3, p_6, p_7\}\times\{q_2, q_3\},
\{p_3, p_6\}\times\{q_7, q_9\},
\{p_5\}\times \{q_8, q_{10}\},
\{p_7\}\times \{q_8,q_{11}\}$. (b) Another matrix $M$, similarly decomposed, where the reachable positions
are marked with an x.}
\label{fig:submatrices}
\end{figure}

We say that a graph $P_{t}\times Q_{t}\in\Gamma$ \emph{intersects}
a row $i$ (resp., a column $j$) in $M$ if $p_{i}\in P_{t}$ (resp., $q_j\in Q_{t}$). We denote the subset of the graphs of $\Gamma$ that
intersect row $i$ of $M$
by $\Gamma_{i}^{r}$ and those that intersect the $j$th column by $\Gamma_j^c$. The sets $\Gamma_i^r$ are easily constructed from the lists $L_{P_t}$ of the graphs in $\Gamma$, and are maintained as linked lists.
Similarly, the sets $\Gamma_j^c$ are constructed from the lists $L_{Q_t}$, and are maintained  as doubly-linked lists, so as to facilitate deletions of elements from them. We have $\sum_{i}|\Gamma_{i}^{r}|=\sum_{t}|P_{t}|=O((m^{2/3}n^{2/3}+m+n)\log n)$
 and
$\sum_{j}|\Gamma_{j}^{c}|=\sum_{t}|Q_{t}|=O((m^{2/3}n^{2/3}+m+n)\log n).$

We define a 1-entry $(p_k, q_j)$
to be {\em reachable from below row $i$}, if $k \geq i$ and there exists an entry
$(p_\ell, q_j)$, $\ell < i$, which is reachable.
We process the rows of $M$ in increasing order
and for each graph $P_{t}\times Q_{t}\in\Gamma$  maintain a
reachability variable $v_t$, which is initially set to $\infty$.
We maintain the invariant that when we
start processing row $i$,
if $P_t\times Q_t$ intersects at least one row $i'\geq i$, then
$v_t$ stores the smallest index $j$
for which there exists an entry $(p_k,q_j) \in P_{t}\times Q_{t}$
that is reachable from below row $i$.

Before we start processing the rows of $M$, we verify that
$M_{0,0}=1$ and $M_{m-1,n-1} = 1$, and abort the computation if this
is not the case, determining that $\dstwoFrechet>\delta$.

Assuming that $M_{0,0}=1$, each position in $U_0=\{(p_{0},q_{l})
\mid M_{0,l}=1\}$ is a reachable position. It follows that for each
graph $P_t\times Q_t \in \Gamma$, $v_t$ should be set to
$\min\{l\mid P_t\times Q_t\in \Gamma_l^c \text{ and } (p_0, q_l)\in
U_0\}$. Note that graphs $P_t\times Q_t$ in this set are not necessarily in $\Gamma_0^r$.
We update the $v_t$'s using this rule, as follows. We
first compute $U_0$, the set of pairs, each consisting of $p_0$ and
an element of the union of the lists $L_{Q_t}$, for $P_t\times
Q_t\in \Gamma_0^r$. Then, for each $(p_0,q_l) \in U_0$, we set, for
each graph $P_u\times Q_u \in \Gamma_l^c$, $v_u \leftarrow
\min\{v_u,l\}$.

In principle, this step should now be repeated for each row $i$.
That is, we should compute $y_i=\min\{v_t\mid P_t\times Q_t\in
\Gamma_i^r\}$; this is the index of the leftmost entry of row $i$
that is reachable from below row $i$. Next, we should compute
$U_i=\{(p_i,q_l)\mid M_{i,l}=1 \text{ and } l\geq y_i\}$ as the
union of those pairs that consist of $p_i$ and an element of
$$
\{ q_j \mid q_j \in L_{Q_t} \;\mbox{for}\; P_t\times Q_t\in
\Gamma_i^r \;\mbox{and}\; j \ge y_i \}.$$

The set $U_i$ is the set of reachable positions in row $i$.
Then we should set for each $(p_0,q_l) \in U_i$ and for each graph
$P_u\times Q_u \in \Gamma_l^c$, $v_u \leftarrow \min\{v_u,l\}$.
This however is too expensive, because it may make us construct
explicitly all the $1$-entries of $M$.

To reduce the cost of this step, we note that, for any graph
$P_t\times Q_t$, as soon as $v_t$ is set to some column $l$ at some
point during processing, we can remove $q_l$ from $L_{Q_t}$ because
its presence in this list has no effect on further updates of the
$v_t$'s. Hence, at each step in which we examine a graph $P_t\times
Q_t \in \Gamma_l^c$, for some column $l$, we remove $q_l$ from
$L_{Q_t}$. This removes $q_l$ from any further consideration in rows
with index greater than $i$ and, in particular, $\Gamma_l^c$ will
not be accessed anymore.  This is done also when processing the
first row.

Specifically, we process the rows in increasing order and when we
process row $i$, we first compute $y_i=\min\{v_t\mid P_{t}\times
Q_{t}\in\Gamma_{i}^{r}\}$, in a straightforward manner. (If $i=0$, then we simply set $y_0=1$.)
Then we
construct a set $U_i' \subseteq U_i$ of the ``relevant'' (i.e., reachable)
$1$-entries in row $i$ as follows. For each graph  $P_t
\times Q_t\in \Gamma_i^r$ we traverse (the current) $L_{Q_t}$
\emph{backwards}, and for each $q_j \in L_{Q_t}$ such that $j \ge y_i$
we add $(p_i,q_j)$ to $U_i'$. Then, for each $(p_i,q_l)\in U_i'$,
we go over all graphs $P_u\times Q_u\in \Gamma_l^c$, and set $v_u
\leftarrow\min\{v_u,l\}$. After doing so, we remove $q_l$ from all
the corresponding lists $L_{Q_u}$.

When we process row $m-1$ (the last row of $M$), we set
$y_{m-1}=\min\{v_t\mid P_{t}\times Q_{t} \in\Gamma_{m-1}^{r}\}$. If
$y_{m-1}<\infty$, we conclude that $\dstwoFrechet\le\delta$
(recalling that we already know that $M_{m-1,n-1}=1$). Otherwise, we
conclude that $\dstwoFrechet>\delta$.

\paragraph{Correctness.}
We need to show that $\dstwoFrechet\le\delta$ if and only if
$y_{m-1} <\infty$ (when we start processing row $m-1$).  To this end, we
establish in Lemma~\ref{lem:two_sided} that the invariant stated above
regarding $v_t$ indeed holds. Hence, if $y_{m-1}<\infty$, then the
position $(p_{m-1},q_{y_{m-1}})$ is reachable from below row $m-1$,
implying that $(p_{m-1},q_{n-1})$ is also a reachable position and thus
$\dstwoFrechet\le\delta$. Conversely, if
$\dstwoFrechet\le\delta$ then $(p_{m-1},q_{n-1})$ is a reachable
position. So, either $(p_{m-1},q_{n-1})$ is reachable from below row $m-1$, or
there exists a position $(p_{m-1}, q_j)$, $j<n-1$, that is reachable from
below row $m-1$ (or both). In either case there exists a graph
$P_t\times Q_t$ in $\Gamma_{m-1}^r$ such that $v_t\leq n-1$ and thus
$y_{m-1} <\infty$.
We next show
that the reachability variables $v_t$ of the graphs in $\Gamma$ are maintained
correctly.
\begin{lemma}
\label{lem:two_sided} For each $i=0,\ldots,m-1$, the following
property holds. Let $P_t \times Q_t$ be a graph in $\Gamma_i^r$, and
let $j$ denote the smallest index for which $(p_i,q_j) \in P_t\times
Q_t$ and $(p_i,q_j)$ is reachable from below row $i$. Then, when we
start processing row $i$, we have $v_t = j$.
\end{lemma}

\begin{proof}
We prove this claim by induction on $i$.  For $i=0$, this claim holds trivially.
We assume then that $i>0$ and that the claim is true for each row $i'<i$,
and show that it also holds for row $i$.

Let $P_t\times Q_t$ be a graph in $\Gamma_i^r$, and let $j$ denote the smallest
index for which there exists a position $(p_i, q_j) \in P_t\times Q_t$ that is reachable from
below row $i$. We need to show that $v_t= j$ when we start processing row $i$.

Since $(p_i,q_j)$ is reachable from below row $i$, there exists a position
$(p_k,q_j)$, with $k<i$, that is reachable, and we let $k_0$ denote the smallest
index for which $(p_{k_0}, q_j)$ is reachable.  Let $P_o\times Q_o$ be the graph
containing $(p_{k_0},q_j)$. We first claim that when we start processing row
$k_0$, $q_j$ was not yet deleted from $L_{Q_o}$ (nor from the corresponding
list of any other graph in $\Gamma_j^c$). Assume to the contrary that $q_j$
was deleted from $L_{Q_o}$ before processing row $k_0$. Then there exists a
row $z<k_0$ such that $(p_z, q_j)\in U_z'$ and we deleted $q_j$ from $L_{Q_o}$
when we processed row $z$. By the last assumption, $(p_z,q_j)$ is a reachable
position. This is a contradiction to $k_0$ being the smallest index for which
$(p_{k_0}, q_j)$ is reachable. (The same argument applies for any other graph,
instead of $P_o\times Q_o$.)

We next show that $v_t \leq j$. Since $(p_{k_0},q_j)\in P_o\times Q_o$,
$P_o\times Q_o \in \Gamma_{k_0}^r \cap \Gamma_j^c$. Since $k_0$ is the
smallest index for which $(p_{k_0}, q_j)$ is reachable, there exists an
index $j_0$, such that $j_0 <j$ and $(p_{k_0}, q_{j_0})$ is reachable from
below row $k_0$. (If $k_0=1$, we use instead the starting placement
$(p_0,q_0)$.) It follows from the induction hypothesis that
$y_{k_0}\leq j_0 <j$. Thus, when we processed row $k_0$ and we went
over $L_{Q_o}$, we encountered $q_j$ (as just argued, $q_j$ was still in
that list), and we consequently updated the reachability variables $v_u$ of each
graph in $\Gamma_j^c$, including our graph $P_t\times Q_t$ to be at most $j$.

(Note that if there is no position in $P_t\times Q_t$ that is reachable
from below row $i$ (i.e., $j=\infty$), we trivially have $v_t \leq \infty$.)

Finally, we show that $v_t=j$. Assume to the contrary that $v_t=j_1<j$ when
we start processing row $i$. Then we have updated $v_t$ to hold $j_1$ when
we processed $q_{j_1}$ at some row $k_1 <i$. So, by the induction hypothesis,
$y_{k_1} \leq j_1$, and thus $(p_{k_1}, q_{j_1})$ is a reachable position.
Moreover, $P_t \times Q_t \in \Gamma_{j_1}^c$, since $v_t$ has been updated
to hold $j_1$ when we processed $q_{j_1}$. It follows that
$(p_i, q_{j_1})\in P_t\times Q_t$. Hence, $(p_i, q_{j_1})$ is reachable
from below row $i$. This is a contradiction to $j$ being the smallest
index such that $(p_i, q_j)$ is reachable from below row $i$. This
establishes the induction step and thus completes the proof of the lemma.
\end{proof}

\paragraph{Running Time.} We first analyze the initialization cost of
the data structure, and then the cost of traversal of the rows
for maintaining the variables $v_t$.

\paragraph{Initialization.}
Constructing $\Gamma$ takes $O((m^{2/3}n^{2/3}+m+n)\log(m+n))$ time.
Sorting the lists $L_{P_t}$ (resp., $L_{Q_t}$) of each $P_{t}\times Q_{t}\in\Gamma$ takes
$O((m^{2/3}n^{2/3}+m+n)\log^2 (m+n))$ time. Constructing the lists
$\Gamma_i^r$ (resp., $\Gamma_j^c$) for each $p_i \in P$ (resp., $q_j
\in Q$) takes time linear in the sum of the sizes of the $P_t$'s and
the $Q_t$'s, which is $O((m^{2/3}n^{2/3}+m+n)\log(m+n))$.

\paragraph{Traversing the rows.}
When we process row $i$ we
first compute $y_i$ by scanning $\Gamma_i^r$. This takes a total of
$O\left(\sum_{i}|\Gamma_{i}^{r}|\right)=O((m^{2/3}n^{2/3}+m+n)\log
n)$ for all rows. Since the lists $L_{Q_t}$ are sorted, the
computation of $U_i'$ is linear in the size of $U_i'$.
For each pair
$(p_i,q_\ell)\in U_i'$ we scan $\Gamma_\ell^c$, which must contain at
least one graph $P_t\times Q_t \in \Gamma$ such that $p_i \in P_t$
(and $q_j\in Q_t$).
For each element $P_t\times Q_t\in \Gamma_\ell^c$ we spend constant
time updating $v_t$ and removing $q_\ell$ from $L_{Q_t}$. It follows
that the total time, over all rows, of computing $U_i'$ and scanning
the lists $\Gamma_\ell^c$ is
$O\left(\sum_{l}|\Gamma_{l}^{c}|\right)=O((m^{2/3}n^{2/3}+m+n)\log n)$.

We conclude that the total running time is
$O((m^{2/3}n^{2/3}+m+n)\log^2 (m+n))$.

\subsection{The optimization procedure}
\label{sec:DFDS2_opt}
We use the above decision procedure for finding the optimum
$\dstwoFrechet$, as follows.  Note that if we increase $\delta$
continuously, the set of $1$-entries of $M$ can only grow, and
this can only happen at a distance between a point of $P$ and a
point of $Q$. We thus perform a binary search over the $mn$
pairwise distances between the pairs of $P\times Q$. In each
step of the search we need to determine the $k$th smallest
pairwise distance $r_k$ in $P\times Q$, for some value of $k$.
We do so by using the distance selection algorithm of
Katz and Sharir~\cite{KS97}, which can easily be adapted to
work for this bichromatic scenario.
We then run the decision
procedure on $r_k$, using its output to guide the binary search.
At the end of this search, we obtain two consecutive critical
distances $\delta_1, \delta_2$ such that
$\delta_1 < \dstwoFrechet \leq \delta_2$, and we can
therefore conclude that $\dstwoFrechet = \delta_2$. The
running time of the distance selection algorithm of \cite{KS97}
is $O((m^{2/3}n^{2/3}+m+n)\log^2(m+n))$, which also holds for the bipartite version that we use.
We thus obtain the following main result of this section.

\begin{theorem}
Given a set $P$ of $m$ points and a set $Q$ of $n$ points in the
plane, we can compute (deterministically) the two-sided discrete Fr\'echet distance
with shortcuts $\dstwoFrechet$, in
$O((m^{2/3}n^{2/3}+m+n)\log^3 (m+n))$ time, using
$O((m^{2/3}n^{2/3}+m+n)\log (m+n))$ space.
\end{theorem}

\section{An efficient algorithm for the semi-continuous Fr\'echet distance with one-sided shortcuts}
\label{sec:semi_continuous}
A \emph{curve} $f \subseteq \mathbb{R}^2$ is a continuous
mapping from $[a,b]$ to $\mathbb{R}^2$, where $a,b \in \mathbb{R}$ and $a<b$.
A \emph{polygonal curve} is a curve $f:[0,n]\rightarrow \mathbb{R}^2$ with $n\in \mathbb{N}$, such that for all $i\in \{0,1,\ldots, n-1\}$ each $f_{\mid[i,i+1]}$ is affine, i.e. $f(i+\lambda)=(1-\lambda)f(i)+\lambda f(i+1)$ for all $\lambda \in [0,1]$. The integer $n$ is called the \emph{length} (number of edges) of $f$. By this definition, the parametrization of $f$ is such that $f(j)$ is a vertex if  $j\in \mathbb{N}$.

Let $P=(p_{0},\ldots,p_{m-1})$ denote a
sequence of $m$ points in the plane, and let $f:[0,n]\rightarrow\mathbb{R}^2$ denote a
polygonal curve with $n$ edges. Consider a person that is
walking along $f$ from its starting endpoint to its final endpoint,
and a frog that is jumping along the sequence $P$ of stones. The
frog is allowed to make shortcuts (i.e., skip stones) as long as it
traverses $P$ in the right (increasing) direction, but the person
must trace the complete curve $f$.
Assuming that the person holds the frog by a leash, the semi-continuous Fr\'echet distance with shortcuts $\scsFrechet$ is the
minimal length of a leash that is
required in order to traverse $f$ and (parts of) $P$ in this manner,
taking the frog and the person from $(p_0, f(0))$ to $(p_{m-1}, f(n))$.

For a given length $\delta>0$, we say that a position $(p_i,f(j))$, for $i\in \{0,1,\ldots,m-1\}, j\in [0,n]\subset \mathbb{R}$, of the frog and the person is a \emph{reachable} position if they can reach $(p_i,f(j))$ in this manner, with a leash of length $\delta$.

We now present an algorithm for computing the semi-continuous Fr\'echet distance with shortcuts $\scsFrechet$.
As usual, we first solve, in Section~\ref{sec:semi_decision}, the corresponding decision problem, and then solve, in Section~\ref{sec:semi_optimization}, the optimization problem.
The decision problem is solved using an extension of the decision procedure of the one-sided discrete case. Then, the optimization problem is solved using the general framework of the algorithm of Section~\ref{sec:tec}, but Algorithm~\ref{sec:tec}.1 is replaced by a simpler random sampling algorithm (Algorithm~\ref{sec:semi_continuous}.1), and Algorithm~\ref{sec:tec}.2 is replaced by an algorithm that applies a close inspection of the more complex critical distances that occur in this case (Algorithm~\ref{sec:semi_continuous}.2).

\subsection{Decision procedure for the semi-continuous Fr\'echet distance with shortcuts}
\label{sec:semi_decision}

\begin{figure}[htb]
 \centering \includegraphics[scale=0.8]{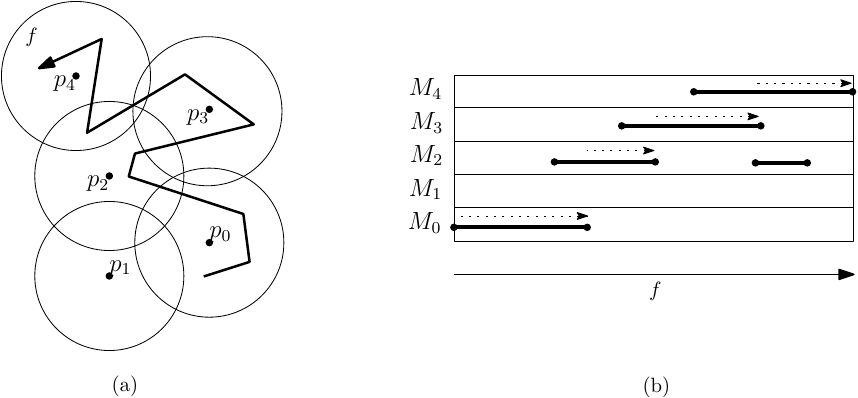}
\caption{\small (a) A sequence of points $P=(p_{0},\ldots,p_{4})$ and a polygonal curve $f\subset\mathbb{R}^2$ with $n=9$ edges.\\
(b) Thinking of $f$ as a continuous mapping from $[0,n]$ to $\mathbb{R}^2$, row $i$ depicts the set $\{t \in [0,n] \mid f(t) \in D_\delta(p_i)\}$. The
dotted subintervals and their connecting upward moves (not drawn) constitute the lowest upward-skipping path between the starting and final positions.}
\label{fig:hdfd}
\end{figure}

Consider the decision version of the problem of the semi-continuous Fr\'echet distance with shortcuts, where, given a
parameter $\delta>0$, we wish to decide whether $\scsFrechet \leq\delta$. This problem can be solved using the decision procedure for
the one-sided DFDS, with a slight modification that takes into account
the continuous nature of $f$. For a
point $p \in \mathbb{R}^2$, $D_\delta(p)$ denotes the disk of
radius $\delta$ centered at $p$.
To visualize the problem, we replace the Boolean matrix $M=M(P,Q)$ of Section~\ref{sec:DFDS1} by a vector $M$ in which the $i$th entry, $M_i$, correspond to $p_i$ and equals
$$M_i=\{t\in[0,n]\mid f(t)\in D_\delta(p_i)\}$$
(see Figure~\ref{fig:hdfd}). That is, each $M_i$ is a finite union of subintervals of $f$.

To decide whether $\delta^-(P,f)\leq \delta$ we need to decide whether there is a monotone ``path'' in $M$ from $(p_0,f(0))$ to $(p_{m-1}, f(n))$ that hops from a subinterval of $M_i$ to a subinterval of $M_j$, $j>i$ only ``over'' a point of $f$ which is in $M_i\cap M_j$. Specifically, we want to decide whether there exists a semi-continuous upward-skipping path from $(p_0,f(0))$ to $(p_{m-1}, f(n))$. A \emph{semi-continuous upward-skipping path} is an alternating sequence of discrete skipping upward moves and continuous right moves. A \emph{discrete skipping upward move} is a move from a reachable position $(p_i,x)$ of the frog and the person to another position $(p_j,x)$ such that $j>i$ and $x\in D_\delta(p_j)$. A \emph{continuous right move} is a move from a reachable position $(p_i,x)$ of the frog and the person to another position $(p_i,x')$, such that the entire portion between $x$ and $x'$ (including $x'$) is contained in $D_\delta(p_i)$. It is easy to verify that there exists a semi-continuous upward-skipping path that reaches $(p_{m-1}, f(n))$ if and only if $\delta^-(P,f)\leq \delta$.

As in the discrete case our decision procedure looks for a ``lowest'' possible upward-skipping path. In this path we first move ``right'' along $f$ as long as we can within the current disk (using the primitive {\bf NextEndpoint} defined below), and then we move to the first among the following disks that contains the current point of $f$ (using the primitive {\bf NextDisk} defined below).

The primitives {\bf NextEndpoint} and {\bf NextDisk} are defined as follows.

\begin{enumerate}[(i)]
\item \label{enu:1}
{\bf NextEndpoint($x, p_i$):} Given a point $x\in f$ and a point $p_i$ of $P$, such that $x\in D_\delta(p_i)$,
 return the forward endpoint of the connected component of $f\cap D_\delta(p_i)$ that contains $x$. This is as far as the person can walk from $x$ along $f$ while the frog stays put at $p_i$.
\item \label{enu:2}
{\bf NextDisk($x,p_i$):}
Given $x$ and $p_i$, as in (\ref{enu:1}), find the smallest $j>i$ such that $x\in D_\delta(p_j)$,
or report that no such index exists (return $j =\infty$). Here the person stays put at $x$ and the frog jumps to the next allowable point (taking a shortcut if needed).
\end{enumerate}

Both primitives admit
efficient implementations. For our purposes it is sufficient to
implement Primitive (\ref{enu:1}) by traversing the edges of $f$ one
by one, starting from the edge containing $x$, and checking for each
such edge $e_j:=f(j)f(j+1)$ of $f$ whether the forward endpoint $f(j+1)$ of $e_j$ belongs
to $D_\delta(p_i)$. For the first edge $e_j$ for which this test fails,
we return the forward endpoint of the interval $e_j \cap D_\delta(p_i)$. It is
also sufficient to implement Primitive (\ref{enu:2}) by checking for
each $j>i$ in increasing order, whether $x \in D_\delta(p_j)$, and
return the first $j$ for which this holds.

\begin{figure}[htbp]
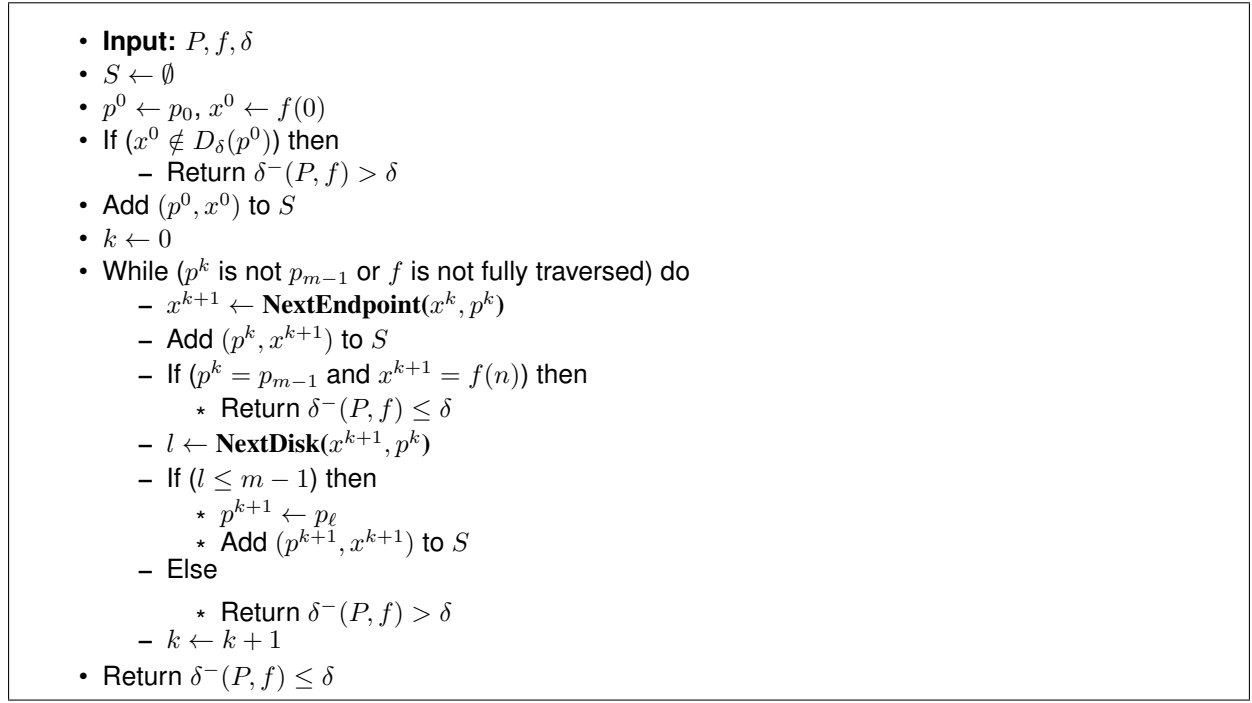

\MyFrame{
\smallskip
\begin{itemize}
{\small \sf
\item \textbf{Input:} $P, f, \delta$ \vspace{-0.3cm}
\item $S\leftarrow \emptyset$ \vspace{-0.3cm}
\item $p^0\leftarrow p_0$, $x^0 \leftarrow f(0)$ \vspace{-0.3cm}
\item If ($x^0 \notin D_\delta(p^0)$) then \vspace{-0.3cm}
    \begin{itemize}
    \item Return $\scsFrechet>\delta$ \vspace{-0.3cm}
    \end{itemize}
\item Add $(p^0,x^0)$ to $S$ \vspace{-0.3cm}
\item $k \leftarrow 0$ \vspace{-0.3cm}
\item While ($p^k$ is not $p_{m-1}$ or $f$ is not fully traversed) do \vspace{-0.3cm}
    \begin{itemize}
          \item $x^{k+1} \leftarrow$ {\bf NextEndpoint($x^k, p^k$)} \vspace{-0.1cm}
          \item Add $(p^k, x^{k+1})$ to $S$ \vspace{-0.1cm}
          \item If ($p^k=p_{m-1}$ and $x^{k+1}=f(n)$) then \vspace{-0.1cm}
          \begin{itemize}
               \item Return $\scsFrechet\leq\delta$ \vspace{-0.1cm}
          \end{itemize}
          \item $l \leftarrow$ {\bf NextDisk($x^{k+1}, p^k$)} \vspace{-0.1cm}
          \item If ($ l \leq m-1$) then \vspace{-0.1cm}
          \begin{itemize}
            \item $p^{k+1} \leftarrow p_\ell$ \vspace{-0.1cm}
            \item Add $(p^{k+1}, x^{k+1})$ to $S$ \vspace{-0.2cm}
          \end{itemize}
          \item Else
            \begin{itemize}
            \item Return $\scsFrechet>\delta$ \vspace{-0.2cm}
            \end{itemize}
          \item $k \leftarrow k+1$ \vspace{-0.3cm}
    \end{itemize}
\item Return $\scsFrechet\leq\delta$
}
\end{itemize}
}\caption{The decision procedure $\Gamma$ for the semi-continuous Fr\'echet distance with shortcuts.} \label{alg:semi-continuous}
\end{figure}

The decision procedure $\Gamma$ itself is given in Figure~\ref{alg:semi-continuous}. $\Gamma$ computes a path $S$ which is a sequence of reachable positions $(p^0,x^0),(p^1,x^1),\ldots$, where $p^k$ is a point of $P$ and $x^k$ is a point on an edge of $f$. The transition from $(p^i,x^i)$ to $(p^{i+1},x^{i+1})$ is either a skipping upward move (if $x^i=x^{i+1}$) or a continuous right move (if $p^i=p^{i+1}$). We denote by $\Pi$ the sequence of pairs $(p^0,s^0),(p^1,s^1),\ldots$ where $s^k$ is either the edge of $f$ containing $x^k$ in its interior, or $x^k$ itself when $x^k$ is a vertex of $f$.\footnote{%
  Note that we use superscripts as in $p^k$ and $s^k$ to denote the
  sequence $S$ defining the solution produced by the decision
  procedure. This is to distinguish them from $p_k$
and $e_k$, with subscripts, that denote the original input sequence
of points for the frog and the sequence of segments of $f$.}

  The correctness of the decision procedure is proved as the correctness of the decision procedure of the one-sided DFDS (of Figure~\ref{alg:semi-sparse}). Specifically, we prove by induction on the steps of the decision procedure, that if there exists a semi-continuous upward-skipping path $S'$ that reaches $(p_{m-1},f(n))$, then the decision procedure maintains a partial semi-continuous path $S$ that is ``below'' $S'$ in the sense that for each $x\in f$, the positions $(p_i,x)\in S$ and $(p_j,x)$ in $S'$ always satisfy $i\leq j$. We omit the details of this proof.

  The running time of this decision procedure is
$O(m+n)$ since we advance along $f$ at each step of Primitive
(\ref{enu:1}), and we advance along $P$ at each step of Primitive
(\ref{enu:2}).

We thus obtain the following lemma.

\begin{lemma}
Given a set $P$
of $m$ points in the plane, a polygonal curve $f$ with $n$ edges in the plane, and a parameter $\delta>0$, we can
determine whether the semi-continuous Fr\'echet distance
$\scsFrechet$ with shortcuts between $P$ and $f$ is at most
$\delta$, in $O(m+n)$ time, using $O(m+n)$ space.
\end{lemma}

We remark that we make no attempt to distinguish between the cases $\delta^-(P,f)<\delta$ and $\delta^-(P,f)=\delta$. This will be taken care of in the optimization procedure, presented next.

\subsection{The optimization procedure}
\label{sec:semi_optimization}
We now use the decision procedure $\Gamma$ to find the optimal value
$\scsFrechet$. To make the dependence on $\delta$ explicit,
we denote, in what follows, the decision procedure for distance
$\delta$ by $\Gamma(\delta)$. The path $S$ computed by
$\Gamma(\delta)$, and each element $(p^k,x^k)$ of $S$, depend on
$\delta$, so we denote them by $S(\delta), p^k(\delta)$ and
$x^k(\delta)$, respectively. The sequence $\Pi$ of pairs $(p^k, s^k)$,
and each of its elements, also depend on $\delta$,
so we denote $\Pi$ by $\Pi(\delta)$, and $s^k$ by $s^k(\delta)$.
Of course, $\Gamma(\delta)$ might fail, i.e., report that
$\scsFrechet > \delta$. In such a case, the path $S(\delta)$
and the sequence of pairs $\Pi(\delta)$ consist of everything that
was accumulated in them until $\Gamma(\delta)$ has terminated (that is, aborted).
In particular, $S(\delta)$ does not end in this case at $(p_{m-1}, f(n))$.

The path $S(\delta_1)$ is {\em combinatorially different} from the
path $S(\delta_2)$, for $\delta_1, \delta_2>0$, if $\Pi(\delta_1) \neq
\Pi(\delta_2)$; otherwise, we say that $S(\delta_1)$ and
$S(\delta_2)$ are {\em combinatorially equivalent}.

For two points $a,b\in \mathbb{R}^2$, the bisector of $a$ and $b$, denoted by $h(a,b)$, is the line containing all the points $\mathbb{R}^2$ that are at equal distance from $a$ and from $b$.

We next argue that each critical value of $\delta$ where $S(\delta)$
changes combinatorially must be of one of the following two types:

\begin{enumerate}
\item\label{enum:pv}
The distance between a point of $P$ and a vertex of $f$ (\emph{point-vertex critical value}).

\item\label{enum:ppe}
For two points $a,b  \in P$ and an edge $e$ of $f$, the distance between $a$ (or $b$)
and the intersection of $e$ with the bisector $h(a,b)$ (\emph{point-point-edge critical value}).
\end{enumerate}
See Figure~\ref{fig:critical} for an illustration.
We assume general position of the input, so as to ensure that these
critical distances are all distinct.

\begin{figure}[htb]
\centering \includegraphics[scale=0.8]{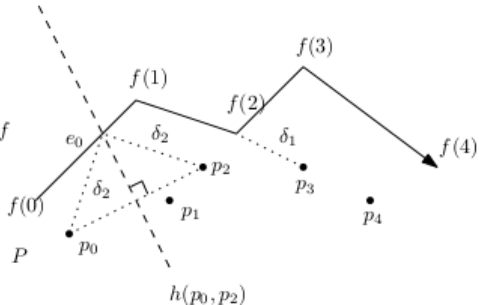}
\caption{\small Two of the critical distances between $f$ and $P$, where $\delta_1$ is a point-vertex distance between $p_3$ and $f(2)$, and $\delta_2$ is a point-point-edge distance between $p_0$, $p_2$ and $e_0$.}
\label{fig:critical}
\end{figure}

\begin{lemma}
Let $\delta$ be such that either $S(\delta^-)$ is combinatorially different
from $S(\delta)$, for all $\delta^- < \delta$ arbitrarily close to
$\delta$, or $S(\delta^+)$ is combinatorially different from $S(\delta)$, for all $\delta^+>\delta$ arbitrarily close to $\delta$. Then $\delta$ is either a point-vertex distance or a
point-point-edge distance.
\end{lemma}
\begin{proof}
In what follows, we use $\delta^-$ and $\delta^+$ to denote an arbitrary point from the respective neighborhood of $\delta$ mentioned in the lemma. Consider the point at which the executions of $\Gamma(\delta^-)$ and of $\Gamma(\delta)$ add a pair to $\Pi(\delta^-)$ which is
different from the pair added to $\Pi(\delta)$ (this includes the case
in which we add a pair to only one of the sets $\Pi(\delta^-)$, $\Pi(\delta)$).

If $(p_0,f(0))$
is in $\Pi(\delta)$ but not in $\Pi(\delta^-)$ then $\delta$ is the
distance between $p_0$ and $f(0)$, a point-vertex distance.

Otherwise, assume that the different pairs arose following a call
to {\bf NextEndPoint($x^k,p^k$)}. Then the edge containing or the vertex coinciding with
$x^{k+1}(\delta) =$ {\bf NextEndPoint($x^k(\delta),p^k(\delta)$)} and
the edge containing or the vertex coinciding with $x^{k+1}(\delta^-) =$ {\bf NextEndPoint($x^k(\delta^-),p^k(\delta^-)$)}
are distinct. Note that $p:=p^k(\delta) =
p^k(\delta^-)$, $s:=s^k(\delta)=s^k(\delta^-)$, and $x:=x^k(\delta)=\Lim{\delta^-\uparrow\delta}x^k(\delta^-)$, since this is the first call that causes a
discrepancy between $\Pi(\delta)$ and $\Pi(\delta^-)$.
Then $D_{\delta^-}(p)\subset D_\delta(p)$, and $x$ belongs to both disks, so {\bf NextEndPoint($x,p$)} terminates at $\delta^-$ at a point $x^{k+1}(\delta^-)$ that precedes its termination point $x^{k+1}(\delta)$ at $\delta$, and converges to $x^{k+1}(\delta)$ as $\delta^-\uparrow \delta$. Since $s^{k+1}(\delta^-)\neq s^{k+1}(\delta)$, the latter must be a vertex of $f$, and $\delta=|px^{k+1}(\delta)|$ is a point-vertex critical distance. A fully symmetric argument yields the same implication when $S(\delta^+)$ is combinatorially different from $S(\delta)$ and the first difference occurs at a call to {\bf NextEndPoint}. (A critical distance between a point of $P$ and (the interior of) an edge of $f$ is irrelevant, since it corresponds to an isolated tangency that cannot be a criticality of a tracing of $f$.)

Finally, assume that the first difference in the pairs added to
$\Pi(\delta^-)$ and $\Pi(\delta)$ occurred following a call to {\bf NextDisk}($x^{k+1}, p^k$).
Put $p^{k+1}(\delta) =$ {\bf NextDisk}$(x^{k+1}(\delta),p^k(\delta))$ and
$p^{k+1}(\delta^-) = $ {\bf NextDisk} $(x^{k+1}(\delta^-),$ $p^k(\delta^-))$.
As before, $p:=p^k(\delta) = p^k(\delta^-)$ by our assumption, $D_{\delta^-}(p)\subset D_\delta(p)$, and, by construction, $x^{k+1}(\delta^-)$ lies on $\bd{D_{\delta^-}(p)}$ and $x^{k+1}(\delta)$ lies on $\bd{D_\delta(p)}$.
Moreover, since $x^{k+1}(\delta)$ is not a vertex of $f$ (or else the previous call to {\bf NextEndPoint} would have produced different pairs at $\delta^-$ and at $\delta$), a simple continuity argument shows that $x^{k+1}(\delta^-)\rightarrow x^{k+1}(\delta) $ as $\delta^-\uparrow \delta$. By assumption, $p^{k+1}(\delta)$ is different from $p^{k+1}(\delta^-)$. We argue, as follows, that in this case $x^{k+1}(\delta)$ must lie on $\bd D_\delta(p^{k}(\delta))$ (this has already been noted), and on $\bd{D_\delta(p^{k+1}(\delta))}$, showing that $\delta$ is a point-point-edge critical distance. Indeed, since $p^{k+1}(\delta)$ is different from $p^{k+1}(\delta^-)$, either $x^{k+1}(\delta)\in D_\delta(p^{k+1}(\delta))$ and $x^{k+1}(\delta^-)\notin D_{\delta^-}(p^{k+1}(\delta))$, or $x^{k+1}(\delta^-)\in D_{\delta^-}(p^{k+1}(\delta^-))$ and $x^{k+1}(\delta)\notin D_{\delta}(p^{k+1}(\delta^-))$, and the latter case is not possible since disks are closed. Again, a fully symmetric argument yields the same conclusion when $S(\delta^+)$ is combinatorially different from $S(\delta)$ and the first difference occurs at a call to {\bf NextDisk}.
\end{proof}

Note that not all triples
of two points $a,b$ of $P$ and an edge $e$ of $f$ necessarily create a point-point-edge
critical event, since the bisector $h(a,b)$ might not intersect $e$.

The following sections develop, using the decision procedure
given above, an algorithm for the optimization problem that runs in
$O((m+n)^{2/3}m^{2/3}n^{1/3}\log(m+n))$ time in expectation and with high probability.
Our algorithm is based on the following two independent building blocks:

\paragraph{Algorithm~\ref{sec:semi_continuous}.1}
An algorithm that finds an interval $(\alpha, \beta]$ that contains, with high probability, $O(L)$ critical distances including $\scsFrechet$, for a given parameter $L\geq 1$. The algorithm runs in $O(m^2n\log(m+n)/L+(m+n)\log(m+n))$
time in expectation and with high probability.

\paragraph{Algorithm~\ref{sec:semi_continuous}.2}
An algorithm that searches for $\scsFrechet$ in $(\alpha, \beta]$ by simulating the decision procedure in an efficient manner. As in Algorithm~\ref{sec:tec}.2, we use the fact that (with high probability) the simulation encounters only $O(L)$ critical distances (as a result of Algorithm~\ref{sec:semi_continuous}.1). This algorithm runs in $O((m+n)L^{1/2}\log(m+n))$ time.

Choosing
$L=m^{4/3}n^{2/3}/(m+n)^{2/3}$, we obtain an algorithm that runs in
$O((m+n)^{2/3}m^{2/3}n^{1/3}$ $\log(m+n))$ time in expectation and with high probability (note that the second term in the bound of Algorithm~\ref{sec:semi_continuous}.1 is always subsumed by this bound). Note that Algorithm~\ref{sec:semi_continuous}.1 (described in Section~\ref{sec:semi_finding_interval}) is different from the analogous algorithm of the discrete case (Algorithm~\ref{sec:tec}.1), and uses
a generalization of a random sampling technique of \cite{HR13}. 
Algorithm~\ref{sec:semi_continuous}.2
(described in Section~\ref{sec:semi_searching_in_interval}) is similar to, but more involved
than, the analogous algorithm of the discrete case (Algorithm~\ref{sec:tec}.2).

\subsubsection{Algorithm~\ref{sec:semi_continuous}.1: Finding an interval that contains $O(L)$ critical distances}
\label{sec:semi_finding_interval}

\begin{lemma}
\label{lem:semi_cont_finding_interval}
Given a polygonal curve $f$ with $n$ edges and a set $P$ of $m$ points
in the plane, and a parameter $L\geq 1$, we can find an interval
$(\alpha,\beta]$ that contains, with high probability, at most $O(L)$
critical distances $\delta$, including $\scsFrechet$,
in $O(m^2n\log(m+n)/L+(m+n)\log(m+n))$ time.
\end{lemma}
\begin{proof}
We generate a random sample $R$ of $cx$ triples of two points of $P$ and
an edge of $f$, where $x=m^2n\log(m+n)/L$, and $c>1$ is
a sufficiently large constant. We also sample $cx$ pairs of a
point of $P$ and a vertex of $f$. This generates at most $2cx$
critical values of $\delta$ (some of the triples that we sample might
not contribute a critical value, as noted above, and are discarded).

We search over the sampled critical values, using the decision procedure $\Gamma$, to find two consecutive values $\alpha, \beta$ of $R$ such that $\scsFrechet\in (\alpha,\beta]$.
This is done in $O(m^2n\log(m+n)/L+(m+n)\log(m+n))$ time, using a
linear time median finding algorithm.

We claim that the interval $(\alpha,\beta]$ that this procedure generates contains,
with high probability, $O(L)$ (non-sampled) critical values of
$\delta$, including $\scsFrechet$. 
To see that, consider the set
$U$ of the $L/2$ (defined) critical values that are smaller than
$\scsFrechet$ and closest to it, and denote by $u$ (resp., $v$) the number of point-vertex distances (resp., point-point-edge distances) among them; thus $u+v = L/2$. (The analysis assumes that there are at least $L/2$ critical values that precede $\scsFrechet$; the argument becomes vacuous when there are fewer such values.) The probability that none of the $2cx$
triples and pairs that we sampled generate a critical value in $U$ is
at most\footnote{Here we assume the model where we make $cx$ independent draws of a point and a vertex, and $cx$ independent draws of two points and an edge. Other alternative sampling models yield similarly small failure probabilities.}
$$
\left(1-\dfrac{u}{{m \choose 2}n}\right)^{cx}\cdot
\left(1-\dfrac{v}{mn}\right)^{cx} \leq e^{-cx\left(\frac{u}{{m \choose 2} n}+\frac{v}{mn}\right)} \leq e^{-cx \cdot\frac{u+v}{m^2n}} = e^{-\frac{c}{2} \log(m+n)} = \dfrac{1}{(m+n)^{c'}},
$$
for some constant $c'$ proportional to $c$.
The same argument, with the same resulting probability bound,
applies for the set $U'$ of the $L/2$ critical values
that are greater than $\scsFrechet$ and closest to it.
Hence, the probability that we miss all the $L/2$ critical values in $U$ and all the $L/2$ critical values in $U'$ is polynomially small.
\end{proof}

\subsubsection{Algorithm~\ref{sec:semi_continuous}.2: An efficient simulation of the decision procedure}
\label{sec:semi_searching_in_interval}

In this section, we show that we can find $\scsFrechet$,
within $(\alpha,\beta]$, in $O((m+n)L^{1/2}\log(m+n))$ time, using a simulation of the decision procedure. Notice the high-level similarity with the discrete counterpart of this algorithm in Section~\ref{sec:searching_in_interval}.

For an edge $e$ of $f$ and two points $a,b\in e$, let $e(a,b]$ be the
subedge of $e$ starting at $a$ (not including $a$) and ending at
$b$; define $e(a,b)$, $e[a,b)$, and $e[a,b]$ in a similar manner. Let $\ell(e)$ denote the line containing $e$.

We simulate the decision procedure $\Gamma$ at the unknown value
$\delta^* = \scsFrechet$.
 Each step of $\Gamma$ involves a call to
one of the procedures {\bf NextEndPoint} and {\bf NextDisk}.
The execution of each of these procedures consists of a sequence
of tests---the former procedure tests the current disk $D_\delta(p^k(\delta))$ against a
sequence of edges of $f$, for finding the first exit point from
the disk, and the latter procedure tests the current point
$x^{k+1}(\delta)$ against a sequence of disks centered at the points
of $P$, for finding the first disk (following the present one) that contains $x^{k+1}(\delta)$.
Each such test generates one or several critical values $\delta$, and we check
whether all these values of $\delta$ lie outside $(\alpha,\beta]$, in which case we know
the (combinatorial nature of the) outcome of the test,
and
we can proceed to the next test. If any of these values $\delta$ lies in $(\alpha,\beta]$,
we bifurcate, proceeding along two branches,
where one assumes that $\delta^* \le \delta$
and the other assumes that $\delta^* > \delta$, or one assumes that $\delta^*<\delta$ and the other assumes $\delta^*\geq \delta$.

These bifurcations generate a tree $T$.
For simplicity of presentation, we represent a single cycle of the
decision procedure (consisting of a call to {\bf NextEndPoint}
followed by a call to {\bf NextDisk}) by two consecutive levels of $T$, each catering to the corresponding call. In Section~\ref{sec:constructT}, we describe the data that we store at the nodes of $T$. Let $v$ be a node of $T$ that represents the situation at the beginning
of such a cycle.
 We show how to simulate a call to {\bf NextEndPoint}, that generates the children of $v$, such that the data that we store at these children is correctly maintained. We also classify the critical values that we encounter in this simulation. Next, we show how to simulate a call to {\bf NextDisk}, that generates the grandchildren of $v$, such that the data that we store at the grandchildren of $v$ is also maintained correctly. Here too, we classify the critical values that we encounter in this simulation. Finally, in Section~\ref{sec:phases}, we show how to partition the simulation into phases, similar to the phases of Lemma~\ref{lem:searching_in_interval}, so as to optimize the performance of the algorithm, and obtain the final result of this section.

\subsubsection{The data stored at $T$}
\label{sec:constructT}
At each node of $T$
we maintain a unique triple $(\tau, p^k, e^k(\tau))$, where
$\tau$ is a range of possible values for $\delta^*$ which can be open or closed at either endpoint; that is, $\tau$ is of one of the forms $(\alpha,\beta), (\alpha,\beta], [\alpha,\beta),[\alpha,\beta]$ ($\tau$ is in general a subrange of the original range $(\alpha,\beta]$, but, for convenience, we denote it using the same symbols), and where $e^k(\tau)$ is a subsegment of some edge $e^k$ of $f$, with open/closed sides matching those of $\tau$.
Each such triple satisfies the following invariant. \begin{quote} (i) For each
$\delta \in \tau$ there exists a pair $(p^k,x^k(\delta))\in
S(\delta)$ such that $p^k(\delta)=p^k$, and \newline (ii) $e^k(\tau)$
is the set of all points $x^k(\delta)$, for $\delta\in\tau$.\end{quote} In
particular if, say, $\tau=(\alpha,\beta]$ then the endpoints $a$ and $b$ of $e^k(\tau)$ are such that $b=x^k(\beta)$ and $a$ is the limit of $x^k(\alpha^+)$
where $\alpha^+$ approaches $\alpha$ from above; similar correspondences occur in all the three other cases. We also maintain the invariant that, for any subtree $T'$ of $T$, the ranges $\tau$ stored at the leaves of $T'$ are pairwise disjoint, and their union is the range stored at the root of $T'$.

The process is initialized as follows. We place the frog at $p_0$, and find the corresponding segment $e^0$ as follows. We compute the distances $|p_0f(0)|, |p_0f(1)|,\ldots,|p_0f(n)|$ from $p_0$ to all the vertices of $f$, and run a binary search through them, or rather through the subset of these values that are within $(\alpha,\beta]$, the interval provided by Algorithm~\ref{sec:semi_continuous}.1, using the decision procedure. This narrows $(\alpha,\beta]$ down to a potentially smaller interval that contains $\delta^*$. As already mentioned, for convenience, we denote this smaller interval also as $(\alpha,\beta]$. We call {\bf NextEndPoint}$(f(0),p_0)$ at $\alpha$ and at $\beta$, and obtain two respective exit points $x^1(\alpha), x^1(\beta)$, that lie on respective edges $e^1(\alpha),e^1(\beta)$, where, for notational convenience, $e^1(\alpha)$ denotes either the relative interior of an edge of $f$, or a vertex of $f$, and similarly for $e^1(\beta)$.
If $e^1(\alpha)=e^1(\beta)$, we store at the root the triple $((\alpha,\beta],p_0, e((\alpha,\beta]))$, where $e$ is the common edge $e^1(\alpha)=e^1(\beta)$. Otherwise, since $\tau=(\alpha,\beta]$ is left-open, $e^1(\alpha)$ must be the relative interior of an edge $e$ and we must have that $\beta=|p_0v|$, where $v$ is the forward endpoint of $e$. Indeed, if $\beta<|p_0v|$ then $x^1(\beta)$ must also be in the relative interior of $e$, as is easily checked, contradicting the assumption that $e^1(\alpha)\neq e^1(\beta)$, and if $\beta>|p_0v|$ then we have $\alpha<|p_0v|<\beta$, contradicting the fact that the preceding binary search through the sequence of all distances from $p_0$ to the vertices of $f$ has ended at $(\alpha,\beta]$. So when $e^1(\alpha)\neq e^1(\beta)$ the root bifurcates into two nodes, one storing $((\alpha,\beta),p_0,e((\alpha,\beta)))$ (where $e(\alpha,\beta)$ ends at $v$ but does not include it), and one storing $([\beta],p_0,\tilde{e}(\beta))$, where $\tilde{e}(\beta)$ is the edge or vertex that contains $x^1(\beta)$: it could be $v$ if $f$ exits $D_\beta(p_0)$ at $v$, or anywhere further along $f$; see Figure~\ref{fig:xx}.
Note that we may assume that $\alpha\geq\alpha_0:=\max\{|p_0f(0)|,|p_{m-1}f(n)|\}$, because $\delta^-(P,f)$ must be at least $\alpha_0$.

\begin{figure}
        \centering
        \begin{subfigure}[b]{0.3\textwidth}
                \centering
                \includegraphics[width=\textwidth]{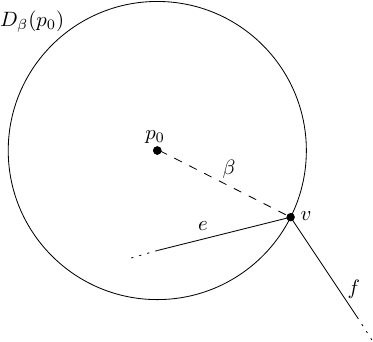}
                (a)
                \label{fig:xx_1}
        \end{subfigure}%
\hspace{1cm}
        \begin{subfigure}[b]{0.3\textwidth}
                \centering
                \includegraphics[width=\textwidth]{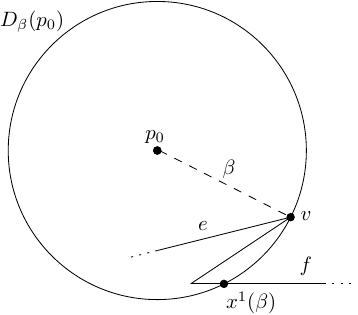}
                (b)
                \label{fig:xx_2}
        \end{subfigure}

                \caption{\small (a) $f$ exits $D_\beta(p_0)$ at $v_\ell$. (b) $f$ exits $D_\beta(p_0)$ further along $f$.}
\label{fig:xx}
\end{figure}

Let $v$ be a node of $T$ that represents the situation at the beginning
of a cycle consisting of a call to {\bf NextEndPoint}
followed by a call to {\bf NextDisk}.  We now show how to construct the triples for the
children and the grandchildren of $v$ from the triple
$(\tau, p^k, e^k(\tau))$ of $v$.

\paragraph{A simulation of {\bf NextEndPoint}.}
To construct the
children of $v$, we simulate {\bf NextEndPoint}, assuming that the current pair in $S$ is $(p^k,x^k(\delta))$, for $\delta \in \tau$, and $x^k(\delta)\in e^k(\tau)$.
Let $\alpha,\beta$ denote the left and right endpoints of $\tau$, respectively.
 The rough, informal idea is to compute, for $\delta=\alpha$ and for $\delta=\beta$, the edge or vertex containing the forward endpoint of the connected component of $f \cap D_\delta(p^k)$ that contains $x^k(\delta)$. If we obtain the same edge or vertex $e$ for $\delta=\alpha$ and for $\delta=\beta$, we conclude that all values in $\tau$ agree that $(p^k, e)$ is the next pair in $\Pi$, and we continue to the next step of the procedure, with a single child of $v$ that stores $(\tau, p^k, e(\tau))$. Otherwise, we have detected at least one point-vertex critical value $\delta_0$ in $\tau$, at which we exit from $f$ at an endpoint of an edge that lies on $\bd{D_{\delta_0}(p^k)}$. We then bifurcate, proceeding along several paths, whose $\delta$-ranges are delimited at the critical distances $\delta_0$.

We now provide a more precise and detailed description of the simulation of {\bf NextEndPoint} at $v$.
Let $a_\alpha$ be the first intersection of $\bd{D_{\alpha}(p^k)}$ with $f$
following $a=x^k(\alpha)$, and let $b_\beta$ be the first intersection
of $\bd{D_{\beta}(p^k)}$ with $f$ following $b=x^k(\beta)$. Let $e_j$ be the
edge of $f$ equal to $e^k$.  We traverse $e_j, e_{j+1},\ldots, e_{n-1}$
in order, as well as the vertices of $f$ delimiting them, and for each such edge $e_\ell$ or vertex $v_\ell=f(\ell)$, we have three possible cases. We first discuss the case of edges, and then handle vertices.

(i) $a_\alpha \notin e_\ell$ and $b_\beta \notin e_\ell$.
In this case, the forward endpoint
$x^{k+1}(\delta)$ of the connected component of $f\cap D_{\delta}(p^k)$
containing $x^k(\delta)$ is not in $e_\ell$, for all
$\delta \in \tau$. So we proceed to the next vertex $v_\ell$ and edge $e_{\ell+1}$.

(ii) ${a_\alpha \in e_\ell}$ and ${b_\beta \in e_\ell}$.
In this case, the set $e_\ell(\tau)$ of all the
forward endpoints $x^{k+1}(\delta)$ of the connected components
of $f \cap {D_{\delta}(p^k)}$ containing $x^{k}(\delta)$, for
$\delta\in\tau$, is the subinterval of $e_\ell$ delimited by $a_\alpha$ and $b_\beta$; it is open/closed at $a_\alpha$ if $\tau$ is open/closed at $\alpha$, and similarly for $b_\beta$ and $\beta$. In this case $v$ has a single child $v'$, with
the triple $(\tau, p^k, e_\ell(\tau))$.

(iii) ${a_\alpha \in e_\ell}$ and ${b_\beta \notin e_\ell}$.
In this case, we encounter a point-vertex critical distance
$\delta_0 \in \tau$, between $p^k$ and $v_\ell=f(\ell)$. That is, for each
$\delta\geq \delta_0 \in \tau$, the forward endpoint of the connected
component of $f\cap D_{\delta}(p^k)$ containing $x^k(\delta)$ is not in
$e_\ell$ (but in an edge or at a vertex following $e_\ell$), and for each
$\delta<\delta_0\in \tau$, the forward endpoint of the
connected component $f\cap D_{\delta}(p^k)$ containing $x^k(\delta)$
is in $e_\ell$, between $a_\alpha$ and $v_\ell$. See Figure~\ref{fig:yy}.
Note that we may not yet have a fixed edge or vertex of $f$ at which we exit from $D_\delta(p^k)$ for all $\delta\geq \delta_0$, and we may have to perform further bifurcations. Nevertheless, the situation is clear for $\delta<\delta_0$: we generate a child $v'$ of $v$ that stores the triple $(\tau^-, p^k, e_\ell(\tau^-))$, where $\tau^-=\tau\cap (-\infty, \delta_0)$. We then continue to generate further children of $v$ with the range $\tau^+=\tau\cap[\delta_0,\infty)$ instead of $\tau$. Note that, to proceed, we need to compute $a_{\delta_0}$, which may either be $v_\ell$ itself (in the scenario depicted in Figure~\ref{fig:yy}(a)), or at some further edge or vertex (as in Figure~\ref{fig:yy}(b)).

(iv) Since $a$ precedes $b$ on $e^k$ and $D_\alpha(p^k)\subset D_\beta(p^k)$, the fourth possibility, where $a_\alpha \notin e_\ell$ and $b_\beta \in e_\ell$, is impossible.

\begin{figure}
        \centering
        \begin{subfigure}[b]{0.3\textwidth}
                \centering
                \includegraphics[width=\textwidth]{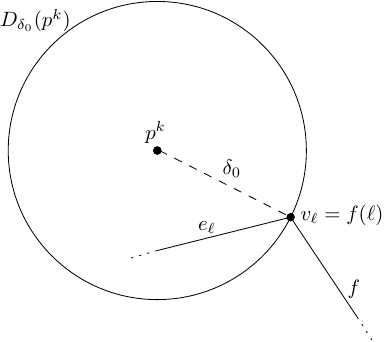}
                (a)
                \label{fig:yy_1}
        \end{subfigure}%
\hspace{1cm}
        \begin{subfigure}[b]{0.3\textwidth}
                \centering
                \includegraphics[width=\textwidth]{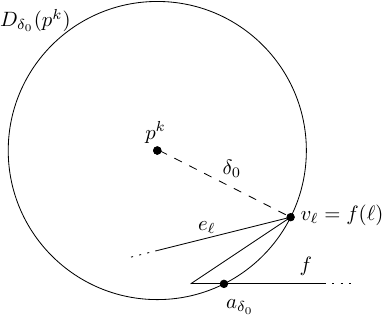}
                (b)
                \label{fig:yy_2}
        \end{subfigure}

                \caption{\small (a) $f$ exits $D_{\delta_0}(p^k)$ at $v$. (b) $f$ exits $D_{\delta_0}(p^k)$ further along $f$.}
\label{fig:yy}
\end{figure}

Handling a vertex $v_\ell=f(\ell)$ is done in a similar but simpler manner. In case (i) we move to the next edge $e_{\ell+1}$, as before. Case (ii) can arise only if $\alpha=\beta$ in which case $v$ has a single child $v'$ that stores the triple
$([\beta],p^k,v_\ell)$. In case (iii) (which arises only in scenarios like those in Figure~\ref{fig:yy}(a)) we bifurcate, we generate a child $v'$ of $v$ that stores the triple $([\delta_0],p^k,v_\ell)$, where $\delta_0=|p^kv_\ell|$, and continue the expansion of $v$ with the range $\tau^+=\tau\cap(\delta_0,\infty)$.

Note that, in either of the cases discussed above, we maintain the invariant that the ranges at the leaves of any subtree are a disjoint cover of the range of the root of the subtree.

\paragraph{A simulation of {\bf NextDisk}.}
Next we generate the grandchildren $v''$  of $v$,
which result from the simulation of the call to {\bf NextDisk}. Let
$(\tau, p^k, e^{k}(\tau))$ be the triple of a child $v'$ of $v$.
Let $\alpha,\beta$ denote the left and right endpoints of $\tau$, respectively.
The idea is to compute, for $\delta=\alpha$ and for $\delta=\beta$, the next point $p_\ell$ of $P$ such that the disk $D_\delta(p_\ell)$ contains $x^k(\delta)$. If we obtain the same point $p_\ell$ for $\delta=\alpha$ and for $\delta=\beta$, we conclude that all values in $\tau$ agree that $(p_\ell, e^k)$ is the next pair in $\Pi$. We add a single child $v''$ of $v'$ that stores the triple
$(\tau,p_\ell,e^{k}(\tau))$ and  continue to the next step of the simulation. Otherwise, let $p_\ell$ be the point returned for $\alpha$, and let $p_{\ell'}$ be the point returned for $\beta$. Then, as we prove in detail below, there must exist $\delta_0 \in \tau$ such that $x^k(\delta_0)$, which lies on $\bd{D_{\delta_0}(p^k)}$, also lies on $\bd{D_{\delta_0}(p_\ell)}$ or $\bd{D_{\delta_0}(p_{\ell'})}$. That is, we have detected a point-point-edge critical value $\delta_0$ in $\tau$, and we bifurcate, proceeding along several paths, delimited at critical distances of this sort.

A precise detailed description of the simulation of {\bf NextDisk} at $v'$ goes as follows. We first need the following easy observation, whose trivial proof is omitted
(see Figure~\ref{fig:bisector} for an illustration).

\begin{observation}\label{obs}
Let $p$ and $q$ be two points in the plane,
and let $s$ be a point on $\bd{D_\delta(p)}$, for some $\delta>0$.
Then $s\in D_\delta(q)$ if and only if $s$ is in the halfspace
bounded by $h(p,q)$ that contains $q$.
\end{observation}

\begin{figure}[htb]
\centering \includegraphics[scale=0.6]{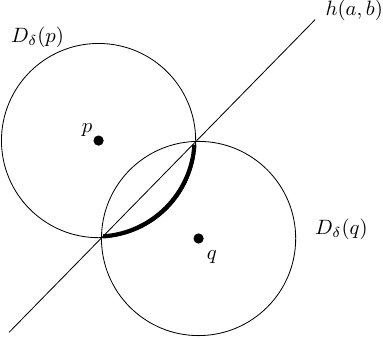}
\centering \caption{\small The points on $\bd{D_\delta(p)}$ that are in
$D_\delta(q)$ are in the halfspace bounded by $h(p,q)$ that contains $q$.}\label{fig:bisector}
\end{figure}

Let $p^k$ be the point $p_i$ of $P$. We
simulate {\bf NextDisk} at all possible $\delta \in \tau$ by traversing
$p_{i+1},\ldots, p_{m-1}$, distinguishing between the following cases at each such
point $p_\ell$.
Let $a$ and $b$ denote the endpoints of $e^k(\tau)$.

(i) ${a \notin D_\alpha(p_{\ell})}$ and ${b \notin D_\beta(p_{\ell})}$.
In this case, each point $x^{k}(\delta)$ on $e^{k}$, for
$\delta \in\tau$, satisfies
$x^{k}(\delta) \notin D_{\delta}(p_\ell)$. Indeed, by the way we
computed the triple for $v'$, $a$ is a point on $\bd{D_\alpha(p^k)}$
and $b$ is a point on $\bd{D_\beta(p^k)}$. Thus, by
Observation~\ref{obs}, $a$ and $b$ are not in the halfspace $h^+(p^k,p_\ell)$ bounded
by $h(p^k,p_\ell)$ that contains $p_\ell$. Thus, $x^{k}(\delta)$, for any
$\delta\in \tau$, is also not in the halfspace $h^+(p^k,p_\ell)$. Since $x^{k}(\delta)$ is a
point on $\bd{D_{\delta}(p^k)}$, again by Observation~\ref{obs},
$x^{k}(\delta) \notin D_{\delta}(p_\ell)$. Hence, in this case, the
$P$-frog cannot jump to $p_\ell$ when the person is at $x^{k}(\delta)$
(for any point $x^{k}(\delta) \in e^{k}(\tau)$), so we proceed to the next point $p_{\ell+1}$.

(ii) ${a \in D_\alpha(p_{\ell})}$ and ${b\in D_\beta(p_{\ell})}$.
By a similar reasoning to that of the preceding case,
each point $x^{k}(\delta) \in e^{k}(\tau)$
satisfies $x^{k}(\delta) \in D_{\delta}(p_\ell)$. Hence, for each
point $x^{k}(\delta)\in e^{k}(\tau)$, the $P$-frog can jump to
$p_{\ell}$ when the person is at $x^{k}(\delta)$.
So in this case $v'$ has only one child $v''$ that corresponds to
the triple $(\tau, p_\ell, e^{k}(\tau))$.


(iii) ${b \in D_\beta(p_{\ell})}$ and ${a \notin D_\alpha(p_{\ell})}$.
By a similar reasoning as in the previous
cases, using Observation~\ref{obs}, $b$ is in the (closed) halfspace
$h^+(p^k,p_\ell)$, and $a$ is not.
Thus, there exists a point $s$ such that $s = e^k(\tau) \cap h(p^k,p_\ell)$.
Put $\tau^-=\tau\cap(-\infty,\delta_0)$ and $\tau^+=\tau\cap[\delta_0, \infty)$, where $\delta_0=|p^k-s|=|p_\ell-s|$.
Note that $\delta_0$
is a point-point-edge
critical value involving $p^k, p_\ell$ and $e^k$.

By construction, if $\delta\in \tau^-$ then $x^k(\delta)\in e^k(\tau^-)$, and by Observation~\ref{obs} $x^k(\delta)\notin D_\delta(p_\ell)$ (so the frog cannot jump to $p_\ell$ when the person is at $x^k(\delta)$). Similarly, if $\delta\in \tau^+$ then $x^k(\delta)\in e^k(\tau^+)$ and $x^k(\delta)\in D_\delta(p_\ell)$ (so the frog can jump to $p_\ell$ when the person is at $x^k(\delta)$). See Figure~\ref{fig:bifurcate}(a).

Consequently, we bifurcate at $\delta_0$. That is, we generate a child $v''$ of $v'$ that corresponds to the triple
$(\tau^+, p_\ell, e^k(\tau^+))$, and continue to generate the
other children of $v'$ by proceeding to the next point (if there is
one) $p_{\ell + 1}$ with the updated triple $(\tau^-,
p^k, e^k(\tau^-))$. Note that in this case the other children of $v'$ will precede $v''$ in the order of their ranges.

(iv) ${a\in D_\alpha(p_{\ell})}$ and ${b \notin D_\beta(p_{\ell})}$.
Arguing similarly to the preceding case,
we encounter a point-point-edge critical value $\delta_0$ involving
$p^k, p_\ell$ and $e^k$, where $\delta_0$ is the distance between
$s = e^k(\tau) \cap h(p^k,p_\ell)$ and $p^k$ (or $p_\ell$). Here, though, $\delta_0$ joins the lower range $\tau^-$ and not the upper range $\tau^+$. That is, we put $\tau^-=\tau\cap(-\infty, \delta_0]$ and $\tau^+=\tau\cap(\delta_0, \infty)$.

We generate a child $v''$ of $v'$ that corresponds to the triple
$(\tau^-, p_\ell, e^k(\tau^-))$, and continue to generate
the other children of $v'$ by proceeding to the next point (if there
is one) $p_{\ell + 1}$ with the updated triple $(\tau^+,
p^k, e^k(\tau^+))$; this time these other children will succeed $v''$ in the range order. See Figure~\ref{fig:bifurcate}(b).

\begin{figure}
        \centering
        \begin{subfigure}[b]{0.4\textwidth}
                \centering
                \includegraphics[width=\textwidth]{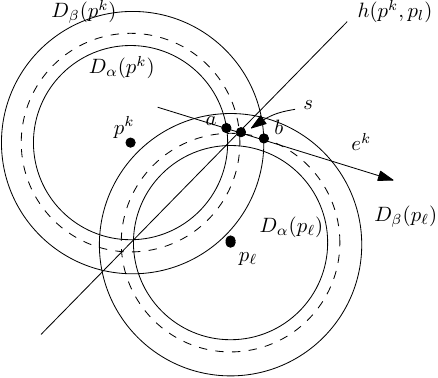}
                (a) $b \in D_\beta(p_{\ell})$ and $a \notin D_\alpha(p_{\ell})$.
                \label{fig:bifurcate1}
        \end{subfigure}%
\hspace{1cm}
        \begin{subfigure}[b]{0.4\textwidth}
                \centering
                \includegraphics[width=\textwidth]{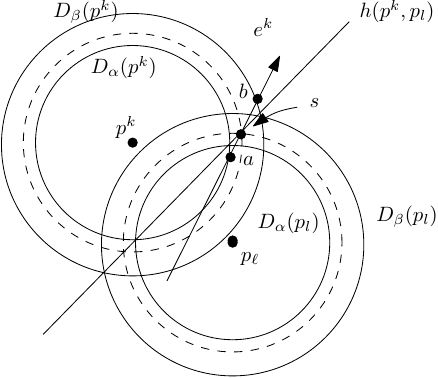}
                (b) $a \in D_\alpha(p_{\ell})$ and $b \notin D_\beta(p_{\ell})$.
                \label{fig:bifurcate2}
        \end{subfigure}

                \caption{\small Situations that cause bifurcation when simulating a call to {\bf NextDisk}.
The disks $D_{\delta_0}(p^k)$ and $D_{\delta_0}(p_{\ell})$, where $\delta_0$ is the corresponding critical value, are drawn dashed. }\label{fig:bifurcate}
\end{figure}

It is straightforward to verify that the triple of each node $v$
that we generate satisfies the invariants mentioned at the beginning
of the proof.

Note that we may reach the last point $p_{m-1}$ without generating
children of $v'$ (i.e., grandchildren of $v$). In this case
$\delta^-(P,f)$ cannot be in $\tau$, and we can abandon this branch of $T_s$ altogether.
When we reach a node whose triple is
$(\tau, p_{m-1}, e_n(\tau))$, and $f(n)$ is the forward endpoint of $e_n(\tau)$ and $e_n(\tau)$ is closed at $f(n)$, then $\delta^-(P,f) = \beta$, assuming that $\delta^-(P,f) \in \tau$.

\subsubsection{The phases in the construction of $T$}
\label{sec:phases}
We do not generate the entire tree $T$ but proceed as follows. Let $s$ be a threshold parameter that we will fix later. We distinguish between unary
nodes $v\in T$, each having a single child, and nodes $v\in T$ with
more than one child. A node $v$ with $d>1$ children is associated
with $d-1$ critical events that triggered these $d-1$
bifurcations. We construct the relevant subtree $T_s$ of $T$ top down, using the simulations of {\bf NextEndPoint} and {\bf NextDisk}. We do not expand a node $v\in T_s$ that has $s$ consecutive unary ancestors immediately preceding it, and we
stop expanding $T_s$ altogether when it contains $m+n$ nodes (we refer to such phases as \emph{unsuccessful}) or when each of
its leaves has $s$ unary ancestors immediately preceding it (these are \emph{successful} phases).
Note that we might stop the construction of $T_s$ in the middle of the expansion of a node $u$ that has too many children. In this case the union of the ranges of the children of $v$ that we generated consists of a prefix and a suffix of the range of v (the suffix is empty when the children are generated at a call to {\bf NextEndPoint},
but both the suffix and the prefix may be nonempty when the children are generated at a call to {\bf NextDisk}).

We then run a binary search over the set of $O(m+n)$ critical values that we have accumulated at the bifurcations of $T_s$, using the decision procedure $\Gamma$ to guide the search. This either identifies a leaf $v$ of $T_s$ such that $\tau_v$ contains $\delta^-(P,f)$, or, if $T_s$ contains some node $u$ which we have not fully expanded, concludes that $\delta^-(P,f)$ is in the part of $\tau_u$ which is not covered by the children of $u$ in $T_s$.
The path of $T_s$ leading to $v$ in the former case and to $u$ in the latter is the next portion of the upward-skipping path $S$ produced by our simulation of the decision procedure (in the algorithm of Figure~\ref{alg:semi-continuous}) at $\delta^-(P,f)$. More precisely, $\tau_v$ (or $\tau_u$ in the latter case) is a subrange of all the nodes on the path, and the portion of $\Pi(\delta^-(P,f))$ encoded along the path is determined (it is fixed for all $\delta\in \tau_v$ (or $\tau_u$)).

We then repeat the whole procedure starting at $v$ (or at $u$, in the second case mentioned above with a reduced range, that excludes the subrange already covered by the children of $u$ in $T_s$). We stop when we reach a node $v$ that records the last step of $S$, which (at $\delta^-(P,f)$) reaches $(p_{m-1},f(n))$. The final range $\tau_v$ of $v$ determines $\delta^-(P,f)$: If $\tau_v=[\alpha,\beta)$ or $[\alpha,\beta]$, we have $\delta^-(P,f)=\alpha$. If $\tau_v=(\alpha,\beta]$, we have $\delta^-(P,f)=\beta$. The fourth case, where $\tau_v=(\alpha,\beta)$, is impossible, as is easily seen. An analysis as in
Lemma~\ref{lem:searching_in_interval} shows that this algorithm
runs in $O((m+n)L^{1/2}\log(m+n))$ time using $O(m+n)$ space.
We thus obtain the following lemma.

\begin{lemma}
\label{lem:semi_cont_searching_in_interval}
Given a polygonal curve $f$ with $n$ edges in the plane, a set
$P$ of $m$ points in the plane, and an interval
$(\alpha,\beta]\subset\reals$ that contains at most $L\geq 1$
critical distances $\delta$ (including $\scsFrechet$),
we can find $\scsFrechet$ in $O((m+n)L^{1/2}\log(m+n))$ time using $O(m+n)$ space.
\end{lemma}

By combining Lemma~\ref{lem:semi_cont_finding_interval} with Lemma~\ref{lem:semi_cont_searching_in_interval}, choosing $L$ to be $m^{4/3}n^{2/3}/(m+n)^{2/3}$, we obtain the following main result of this section.

\begin{theorem}
Given a set $P$ of $m$ points and a polygonal curve $f$ with $n$ edges in the plane, we can compute the one-sided semi-continuous Fr\'echet distance $\scsFrechet$ with shortcuts in $O((m+n)^{2/3}m^{2/3}n^{1/3}\log(m+n))$ time, both in expectation and with high probability, using $O((m+n)^{2/3}m^{2/3}n^{1/3})$ space.
\end{theorem}

\section{Discussion}
The algorithms obtained for the discrete Fr\'echet distance with
shortcuts, run in time significantly better than those for the Fr\'echet
distance without shortcuts. It is thus an interesting open question
whether similar improvements can be obtained for the continuous version of
the Fr\'echet distance with shortcuts, where shortcuts are made only
between vertices of the curves. This variant, that was considered by \cite{DH12}, may be easier than the NP-Hard variant that was
considered by \cite{BDS13}. We hope that the techniques that we have developed for the semi-continuous problem will be useful for tackling this harder problem.

It remains an open question whether the algorithms for the discrete and semi-continuous variants can be further improved. Specifically, it is conceivable that the gap between the linear time decision procedures of the discrete and semi-continuous Fr\'echet distance with one-sided shortcuts and the corresponding optimization procedures can be further reduced. We  conjecture that such an improvement is possible.

In contrast, we are less optimistic concerning the (current approach to) the two-sided variant. The running time of the algorithm for the discrete two-sided variant is based on the running time bound of distance selection between points, where the output is a compact representation of the distances smaller than a specified threshold. A future improved solution of the distance selection problem can be expected to also yield an improvement of the algorithm for the discrete two-sided case. However, in view of similar known lower bounds for related problems (see, e.g.,~\cite{Eri95}), we doubt that the distance selection problem can be solved (significantly) more efficiently.

Another topic for further research is to find additional
applications of some of the ideas that appear in the optimization technique
for the one-sided variants.

\end{document}